\documentclass[fleqn,10pt]{wlscirep}
\usepackage[utf8]{inputenc}
\usepackage[T1]{fontenc}
\usepackage{lmodern}
\usepackage{microtype}
\usepackage{amsmath,amssymb,amsthm}
\usepackage{float}
\usepackage{booktabs}
\usepackage{tabularx}
\usepackage{array}
\usepackage{multirow}
\usepackage{subcaption}
\usepackage{xcolor}
\usepackage[noabbrev,capitalise]{cleveref}
\usepackage{orcidlink}
\graphicspath{{figures/}}

\title{Input-schema identifiability limits in physics-informed surrogates for mechanics-governed flow}

\author[1,*]{Daniel Cie\'slak}
\author[1]{Andrzej Czy\.zewski}

\affil[1]{Gda\'nsk University of Technology, Gda\'nsk, Poland}
\affil[*]{Corresponding author: \href{mailto:daniel.cieslak@pg.edu.pl}{daniel.cieslak@pg.edu.pl}}

\newcommand{\dirmagprobe}{\textsc{direction+magnitude oracle probe}}
\newcommand{\dironly}{\textsc{direction oracle probe}}
\newcommand{\magonly}{\textsc{magnitude oracle probe}}
\newcommand{\geoschema}{\textsc{geometry-only input schema}}
\newcommand{\centerlinegeo}{\textsc{geometry-only input schema (centreline tangent)}}
\newcommand{\sagedirmagprobe}{\textsc{direction+magnitude oracle probe, SAGE backbone}}
\newcommand{\sagegeo}{\textsc{geometry-only input schema, SAGE backbone}}
\newcommand{\wmdir}{\textsc{direction oracle probe (Womersley)}}
\newcommand{\wmwith}{\textsc{direction+magnitude oracle probe (Womersley)}}
\newcommand{\wmnoleak}{\textsc{geometry-only input schema (Womersley)}}
\newcommand{\wmmag}{\textsc{magnitude oracle probe (Womersley)}}
\newcommand{\dmcf}{\textsc{D+M-CF}}
\newcommand{\dcf}{\textsc{D-CF}}
\newcommand{\mcf}{\textsc{M-CF}}
\newcommand{\geocf}{\textsc{GEO}}
\newcommand{\unit}[1]{\hat{\mathbf{#1}}}

\newcommand{\PPdir}[1]{\mathrm{PP_{dir}}@#1^\circ}
\newcommand{\cosgn}{\mathrm{cos}_\mathrm{signed}}
\newcommand{\fracflip}{\phi_\mathrm{flip}}
\newcommand{\ewRMSE}{ewRMSE_{\text{he}}}

\theoremstyle{plain}
\newtheorem{theorem}{Theorem}
\newtheorem{proposition}{Proposition}
\newtheorem{lemma}{Lemma}
\newtheorem{corollary}{Corollary}
\theoremstyle{definition}
\newtheorem{definition}{Definition}
\newtheorem{remark}{Remark}

\begin{abstract}
\noindent Physics-informed and data-driven surrogates are increasingly used to approximate mechanics-governed flow fields, but the target quantities assigned to such models are not always identifiable from the input variables available at prediction time. We introduce an input-schema identifiability certificate for computational surrogates. Starting from a reduced physical model, the certificate decomposes a target field into components that are measurable from geometry, components that require boundary-condition information, and components identifiable only up to a symmetry quotient. This yields a pre-training audit: it predicts which oracle-channel interventions should reduce error, which should fail, and which ambiguity cannot be removed by changing the architecture, loss, optimiser, or sample size. We instantiate the framework for incompressible tubular flow using a Cosserat-rod reduction, where lumen velocity separates into a mesh-measurable tangent direction, a boundary-condition-dependent magnitude, and a signed-orientation ambiguity. Controlled experiments on patient-specific aortic CFD geometries, analytic Womersley flows, and an advection–diffusion transfer problem confirm the predicted pattern: supplying signed direction collapses angular error to the oracle regime, whereas supplying magnitude without orientation leaves the predicted sign ambiguity and yields 16–33\% per-node sign flips. The results provide a mechanics-based diagnostic for deciding whether a surrogate modelling task is physically identifiable before training, and expose failure modes that aggregate error metrics can hide.
\end{abstract}

\begin{document}
\raggedbottom
\maketitle
\thispagestyle{empty}

\section{Introduction}

Computational surrogates are increasingly used as reduced-cost approximations of mechanics-governed simulations in fluid mechanics, biomechanics, and design optimisation. Their accuracy is usually assessed after training, by comparing predicted and simulated fields. A more basic question is often left implicit: whether the requested field is identifiable from the variables supplied to the surrogate at prediction time.

This question is distinct from numerical accuracy, expressive power, or optimisation quality. A highly expressive physics-informed model, graph network, or neural operator cannot recover a boundary-condition-dependent component if the declared input schema contains only geometry. In that case the failure is not a training defect but a well-posedness defect of the surrogate task itself.

Physics-informed machine learning (PIML) has largely treated physics as a constraint on the predictor: conservation laws enter as soft losses~\cite{raissi2019pinn,karniadakis2021pinn}, symmetries enter as hard architectural structure~\cite{satorras2021egnn,suk2024labgatr}, and neural operators learn solution maps for parameterised equations~\cite{li2021fno,kontolati2024neuraloperator}. These tools ask how a model should learn once the prediction task has been declared. We ask the prior question: \emph{is the declared prediction task identifiable from the information the surrogate will actually receive at deployment?}

So formulated, identifiability is a physical well-posedness condition for surrogate learning, not a statistical afterthought. A surrogate input tensor generates a sigma-algebra of observable quantities, and a target field may split into components measurable from it, components requiring boundary-condition information outside it, and components requiring boundary or orientation information not present in the declared input. The theorem we prove constrains the computational question, not the predictor: if a target component is not measurable with respect to the declared input schema, no admissible surrogate can identify it uniformly over the admissible mechanics class without supplying the missing physical information.

That geometry alone cannot determine a boundary-condition-dependent flow is, taken literally, obvious; that is our starting point. The novelty is its operational consequence: a \emph{component-wise input-schema certificate} predicting \emph{before training} which oracle-channel interventions must collapse an error, which must fail, and which residual ambiguity is invariant to architecture, loss, optimiser, and sample size. Instantiated for slender incompressible flow, it classifies the three velocity components as mesh-measurable, boundary-condition-measurable, or identifiable only modulo a discrete orientation quotient; the same construction applies to any target admitting a reduced or asymptotic factorisation, turning a qualitative intuition into a falsifiable protocol that separates a model that learned the flow from one that read a geometric proxy.

\paragraph*{Why this matters now.}
Machine-learning surrogates for computational fluid dynamics (CFD) are now used across the physical sciences~\cite{vinuesa2022mlfluids,cremades2025coherent}, but aggregate error is not a sufficient diagnostic for mechanics-governed vector fields. A model can attain a plausible headline metric by recovering a deterministic geometric proxy for one component while failing on boundary-dependent components. Identifiability therefore belongs alongside conservation, symmetry, discretisation choice, and operator structure as a first-class design condition for computational surrogates. The closest precedent in spirit is Fajardo-Fontiveros \emph{et al.}'s information-theoretic limits on symbolic-model recovery from noisy data~\cite{fajardo2023limits}; we give the PDE-asymptotic analogue---a closed-form bracket on which target components are recoverable from a chosen input set, derived before any predictor is trained.

\paragraph*{Why vascular flow is a clean testbed.}
Cardiovascular flow sharpens the question cleanly. Patient-specific aortas are tubular domains in which geometry, pulsatile inflow, and downstream impedance jointly produce the velocity field, and classical fluid mechanics gives a closed asymptotic reduction of Navier--Stokes that separates the three~\cite{ku1997bloodflow,steinman2002,formaggia20031dbloodflow,sherwin20031dvascular}. A surrogate trained on geometric inputs implicitly claims to recover a function of all three. Recent geometry-conditioned graph surrogates---GEM-GCN~\cite{suk2024gemgcn}, LaB-GATr~\cite{suk2024labgatr}, physics-informed and reduced-order graph networks~\cite{tabe2026pignn,pegolotti2024}, one-dimensional blood-flow networks~\cite{sen2024pignn}, and aneurysm-cohort generalisation~\cite{lannelongue2026npj}---report aggregate accuracy on full input-feature stacks. Taking their accuracy at face value, our contribution is orthogonal: none isolates which input channels drive the prediction, nor states analytically that the prediction \emph{should be possible} from the chosen inputs. Three risk factors of the ML reproducibility crisis~\cite{kapoor2023leakage,alzheimer2025scoping} apply directly---cohorts small by construction in the public-data regime, labels sharing boundary conditions with inputs, and end-to-end metrics that aggregate physically heterogeneous errors---and an identifiability audit is the diagnostic these works lack.

\paragraph*{The Cosserat-rod factorisation.}
We resolve the identifiability question analytically. Modelling a patient-specific vessel as a Cosserat rod~\cite{antman2005nonlinear,formaggia20031dbloodflow} with arc-length $s$, centreline tangent $\unit{T}(s)$ and slowly varying cross-section of radius $R(s)$, and substituting an axisymmetric ansatz into the incompressible Navier--Stokes equations, the leading-order velocity factorises as
\[
\mathbf{u}_\star(s,r,t) = \frac{Q(t)}{\pi R(s)^2}\, f\!\left(\frac{r}{R(s)}\right)\,\unit{T}(s).
\]
This factorisation has three immediate consequences. First, the \emph{direction} field $\unit{\mathbf{u}}_\star = \unit{T}(s)$ is a deterministic functional of mesh geometry, identifiable up to a curvature-induced Dean-number correction. Second, the \emph{magnitude} field is fixed by the inflow flow rate $Q(t)$---a non-geometric boundary-condition quantity~\cite{westerhof2009arterial}---and is not recoverable from geometry alone. Third, a graph predictor that respects the local energy $|\mathbf{u}|^2$ but receives no per-node direction channel suffers an irreducible $\mathbb{Z}_2$ \emph{sign ambiguity} per connected component. These statements are closed-form properties of the velocity field in the lumen; they are independent of any downstream surrogate architecture.

\paragraph*{Empirical audit.}
We test these predictions with controlled oracle-channel interventions on $46$ patient-specific aortas and $34$ analytic Womersley flows, crossing direction $\times$ magnitude availability under an identical FlowGAT model, optimiser, augmentation, and three seeds. The data match the theorem on every prespecified axis: a signed direction probe collapses angular error to the oracle regime; a magnitude probe without direction is non-identifying; the predicted sign ambiguity is realised as a $16$--$33\%$ per-node flip rate; and apparent mass conservation on the patient cohort does not transfer to the analytic regime. The theorem and its matched audit form a computational identifiability protocol for mechanics-governed surrogates, distinguishing a model that learned the vector field from one that exploited a geometric proxy.

\paragraph*{Contributions.}
We make four contributions. First, we formulate \emph{input-schema identifiability} for mechanics-governed surrogate modelling: identifiability is assessed relative to the information available in the prediction-time input schema, not relative to a full PDE problem with all boundary data. Second, we derive a closed-form \emph{pre-training certificate} for slender incompressible tubular flow, using a Cosserat-rod reduction to separate mesh-measurable direction, boundary-condition-dependent magnitude, and a signed-orientation quotient (\Cref{thm:dir-id,prop:quotient,lem:sign}). Third, we convert the certificate into a falsifiable oracle-channel audit with quantitative predictions (P1--P5) stating which interventions must reduce vector-field error, which must fail, and which ambiguity is invariant to architecture, loss, optimiser, regularisation, or training-set size. Fourth, we verify the audit on patient-specific aortic CFD geometries, analytic Womersley flows, architecture and direction-proxy checks, and a non-flow advection--diffusion transfer problem (\Cref{sec:advdiff}); the accompanying repository releases configurations, frozen splits, generators, metrics, bootstrap scripts, and figure/table pipelines (\Cref{sec:data-availability}).

\paragraph*{Scope of the claim.}
The certificate applies to laminar incompressible flow in slender tubular domains that admit a Cosserat-rod reduction; it is not a clinical vascular-flow surrogate, a wall-shear-stress predictor, or a universal theorem for turbulent or strongly non-axisymmetric geometries. The oracle-probe channels are diagnostic interventions that inject selected target-derived quantities to bracket the identifiability ceiling, not deployable inputs. The patient cohort is a realism check supplying paired geometry and CFD velocity; the analytic Womersley benchmark is the mechanism-controlled falsification check because its ground truth is exact. The held-out patient set has $n=5$ cases by construction of the public-data regime, so the evidence is the stable sign and magnitude of the mechanistic contrasts and their analytic replication---not a $p$-value at $n=5$ (\Cref{sec:n5-framing}).

\section{Input-schema identifiability for mechanics-governed surrogates}
\label{sec:theory}

\paragraph*{The audit in seven steps.}
The procedure is stated once, in problem-agnostic form, then instantiated for tubular flow: (1)~define the deployment-time input tensor; (2)~define the target field; (3)~derive or reuse a reduced/asymptotic factorisation of that target; (4)~classify its components against the input tensor as input-measurable, boundary-condition-measurable, quotient-identifiable, or non-identifiable; (5)~design oracle-channel interventions that selectively inject each missing component; (6)~test whether only the identifying channels collapse the corresponding errors while the non-identifying ones leave them unchanged; (7)~verify the obstruction persists under architecture, proxy, or training changes that do not enlarge the input $\sigma$-algebra. Steps 1--4 are analytic and consume no data; steps 5--7 are the falsifiable audit. The Cosserat-rod theorem below is step 3 for slender incompressible flow; \Cref{sec:advdiff} carries the same steps through a synthetic advection--diffusion target to show none is flow-specific.

The analytic claim is about the declared input schema, not the predictor. We formalise the admissible flow class, the four input schemata ($\mathcal{I}_{\mathrm{geo}}$, $\mathcal{I}_{\mathrm{geo+mag}}$, $\mathcal{I}_{\mathrm{geo+dir}}$, $\mathcal{I}_{\mathrm{full}}$) and an $L^2$ tolerance in \Cref{def:input-schema} (Methods), and reduce the lumen flow by a Cosserat-rod expansion factorising the velocity into a geometry-determined direction and a boundary-fixed magnitude, $\mathbf{u}_\star=[Q(t)/\pi R(s)^2]\,f_\alpha(r/R)\,\unit{T}(s)+\mathcal{O}(\epsilon,De)$ (\Cref{eq:closedform}). From this we derive a three-part certificate (proved in \Cref{sec:certificate-formal} and Supplementary \Cref{sec:proofs}): \emph{(i)~direction}---the leading-order tangent line is mesh-measurable up to a Cosserat residual, so direction is $\mathcal{I}_{\mathrm{geo}}$-identifiable; \emph{(ii)~magnitude}---the speed waveform is \emph{not} (admissible inflows share geometry but differ in magnitude), so recovering it requires a boundary channel; \emph{(iii)~sign}---geometry plus a non-negative magnitude channel leaves an irreducible $\mathbb{Z}_2$ sign ambiguity per lumen component unless an asymmetric boundary or signed-orientation channel is supplied, invariant under any predictor class, loss, optimiser or sample size (\Cref{thm:dir-id,prop:quotient,prop:arch-invariance,lem:sign}). The quotient clause yields a lower bound: the geometry-versus-direction angular gap is at least $p^\star\cdot 90^\circ$ minus the Cosserat residual, with $p^\star$ the population-average sign ambiguity (\Cref{cor:gnnaudit}).

\subsection{Falsifiable consequences for surrogate modelling}
\label{sec:falsification}

The certificate makes a risky prediction: if direction is the identifying channel, supplying it must collapse the vector-angle error whether or not magnitude is present, whereas supplying only a non-negative magnitude channel must \emph{not}, since signed orientation remains a quotient variable. A failure of this pattern on the analytic Womersley benchmark---where the tangent direction and exact solution are known by construction---would falsify the audit itself, not merely weaken a neural baseline.

\subsection{Quantitative predictions for the numerical audit}
\label{sec:theory-predictions}

The Cosserat-rod theorem makes five quantitative predictions that we test in \Cref{sec:results-preview} and that structure the empirical sections below.

\textbf{(P1) Symmetry under direction-proxy refinement.} \Cref{thm:dir-id} treats $\unit{T}(s)$ as a property of the mesh, not of the proxy method, so the theorem-predicted gap should be invariant to whether $\unit{T}$ is computed by case-global PCA or by per-node centreline skeletonisation. The \centerlinegeo{} variant should be indistinguishable from \geoschema{}.

\textbf{(P2) Symmetry under architecture choice.} \Cref{lem:sign} concerns the information content of the declared input schema, not the backbone, so the asymmetric pattern should reproduce on a vanilla GraphSAGE backbone (\sagegeo{} vs.\ \sagedirmagprobe{}).

\textbf{(P3) Cross-domain replication.} On a straight cylinder ($\epsilon \to 0$, $De \to 0$), \Cref{thm:dir-id} predicts that the geometric proxy becomes \emph{exact} but \Cref{lem:sign} still gives a non-zero flip rate; the geometry-only-family residual on the Womersley benchmark must therefore persist and be accounted for entirely by sign-degeneracy.

\textbf{(P4) Phase invariance of the flip rate.} \Cref{lem:sign} singles out no preferred sign, so the flip rate should be independent of the cycle phase $\varphi = (\omega t_\mathrm{phase})\bmod 2\pi$ on Womersley. In particular, the forward-flow and reverse-flow flip rates should agree to within sampling noise.

\textbf{(P5) Mass conservation is non-axiomatic for the predictor.} The predictor's continuity behaviour follows from the training-data distribution, not from \Cref{eq:closedform}. On a regime where the data-distributional pattern fails (Womersley analytical truth $\overline{|\nabla\!\cdot\!\mathbf{u}|}\!\to\!0$), the trained predictor should retain its in-vivo divergence floor.

All five predictions are confirmed in \Cref{sec:results-preview,sec:robustness,sec:womersley,sec:continuity} and in Supplementary \Cref{sec:supp_sage,sec:supp_centerline}. This section therefore defines an a-priori, falsifiable analytic frame for the numerical audit that follows.

\begin{table}[H]
\small
\centering
\caption{\textbf{Falsifiable interventions implied by the certificate.}
Each row is a quantitative prediction derived from
\Cref{thm:dir-id,prop:quotient,prop:arch-invariance,lem:sign,cor:gnnaudit}
together with the controlled-input modification that tests it.
The third column is the predicted sign of the effect on the
identifiability bracket; the fourth column is the corresponding
observable in the empirical audit; the fifth column reports the
realised result. A negative outcome in any row would falsify the
corresponding clause of the theorem rather than merely degrade a
neural baseline. The interventions are
$\mathcal{A}_{\mathcal{I}}$-modifications (input-schema or
admissible-class changes), not architecture, loss or
optimiser changes.}
\label{tab:falsifiable}
\begin{tabularx}{\linewidth}{p{0.08\linewidth}p{0.18\linewidth}p{0.20\linewidth}p{0.27\linewidth}X}
\toprule
\# & Intervention on $\mathcal{I}$ or $\mathcal{F}$ &
Predicted effect on $\mathcal{R}^{\star}$ &
Observable &
Realised outcome \\
\midrule
P1 & Replace global PCA tangent with per-node centreline tangent
(refines $\mathcal{G}(\Omega)$, not $\mathcal{A}_{\mathcal{I}}$). &
No change: $\unit{T}$ already $\mathcal{A}_{\mathcal{I}_{\mathrm{geo}}}$-measurable. &
$\cosgn$, $\fracflip$, angle median for \centerlinegeo{} vs \geoschema{}. &
$\cosgn=+0.799$ vs $+0.798$; angle $35.4^\circ$ vs $34.8^\circ$
(Supp.\ \Cref{sec:supp_centerline}). \\
\midrule
P2 & Replace GATv2 attention with vanilla GraphSAGE (changes
$\mathcal{H}_{\mathcal{I}}$ only). &
No change: \Cref{prop:arch-invariance} bracket invariant
under $\mathcal{H}$. &
$\PPdir{10}$, $\fracflip$, angle median for \sagegeo{} vs \geoschema{} and \sagedirmagprobe{} vs \dirmagprobe{}. &
SAGE reproduces both ceiling and floor within seed s.d.\
(Supp.\ \Cref{sec:supp_sage}). \\
\midrule
P3 & Replace patient cohort with straight-tube
analytic Womersley ($\epsilon, De\to 0$). &
Direction bound \Cref{eq:dir-bound} tightens to zero; sign-degeneracy bound \Cref{eq:flip-lower} unchanged. &
$\fracflip$ on geometry-only-schema-family variants. &
$\fracflip=0.33$ on Womersley vs $0.16$ on VMR; direction-oracle-probe variants stay at $\fracflip=0$ (\Cref{sec:womersley}). \\
\midrule
P4 & Bin Womersley test cases by cycle phase $\varphi$;
$Q\mapsto -Q$ is a measure-preserving involution. &
Per-phase flip rate constant: $\pi_+(X)$ independent of $\varphi$. &
$\fracflip$ on forward- vs reverse-flow half-cycle for \wmnoleak{}. &
$0.39$ vs $0.33$, indistinguishable within sampling noise
(\Cref{fig:phase_cos}). \\
\midrule
P5 & Compare divergence statistics
$\overline{|\nabla\!\cdot\!\hat{\mathbf{u}}|}$ across domains
(swaps the distribution of $\nabla\!\cdot Y$, leaves
$\mathcal{A}_{\mathcal{I}}$ intact). &
Predictor inherits training-distribution divergence; no analytic
mass conservation. &
Predicted divergence on VMR vs Womersley. &
Matches CFD discretisation on VMR
($\sim\!2.7\times 10^{-3}$); blows up by $\sim\!10^2$ on Womersley
(\Cref{sec:continuity}). \\
\bottomrule
\end{tabularx}
\end{table}

\subsection{Relation to physically based machine learning and numerical surrogates}
\label{sec:relation-piml}

The certificate is complementary to familiar uses of physics in data-driven modelling. Physics-informed neural networks~\cite{raissi2019pinn,karniadakis2021pinn}, PDE-preserved graph networks~\cite{liu2024multiresolution,sharma2025dynamical}, equivariant backbones~\cite{satorras2021egnn,suk2024labgatr}, and neural operators~\cite{li2021fno,kontolati2024neuraloperator} all use physics to \emph{constrain the predictor}; the Cosserat-rod theorem instead uses physics to \emph{constrain the question}---given a fixed predictor class, what is identifiable from the declared inputs and what is not (developed in \Cref{sec:piml-implications}).

\section{Results}

\begin{figure}[H]
\centering
\includegraphics[width=0.99\linewidth]{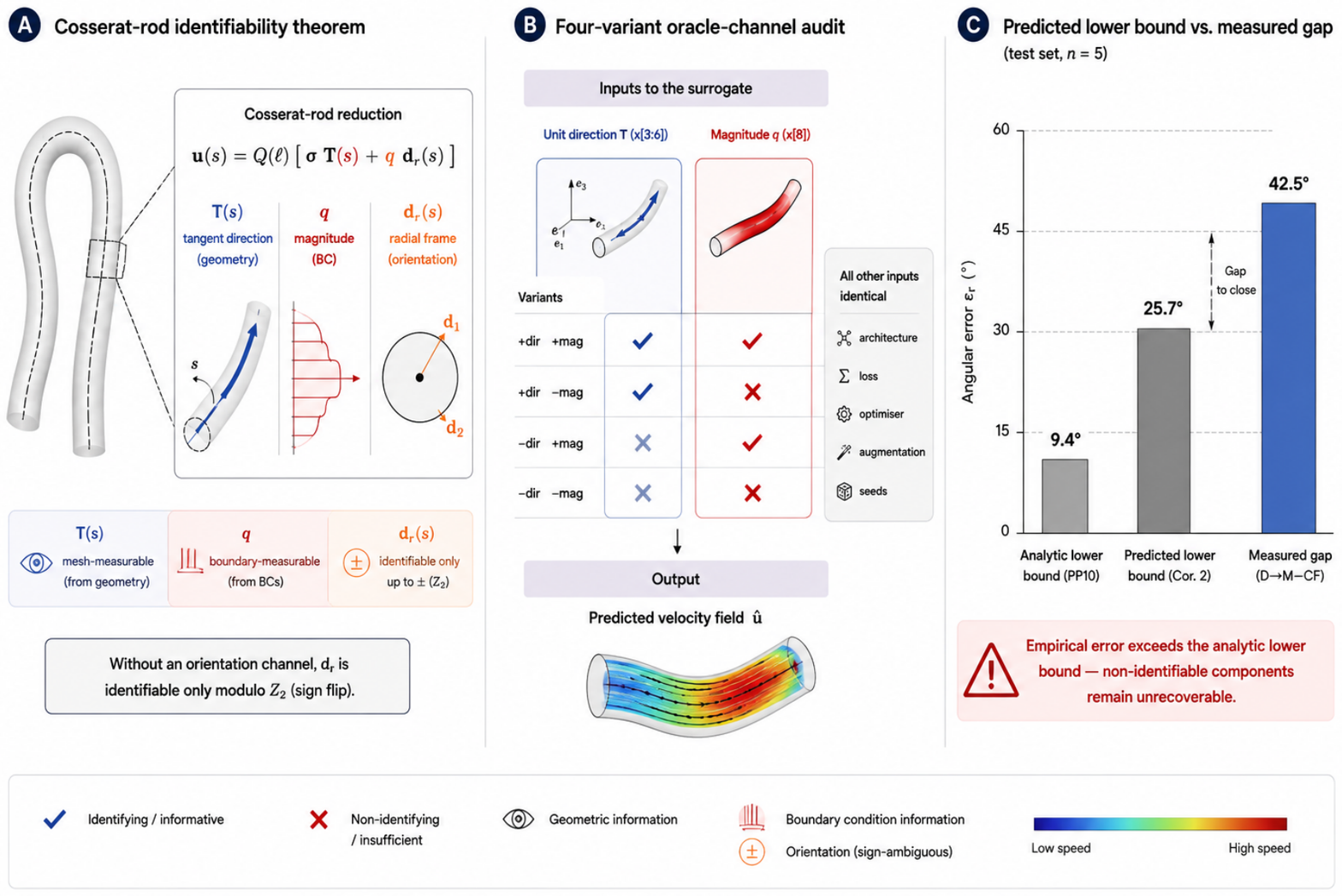}
\caption{\textbf{Cosserat-rod identifiability
theorem and its empirical audit at a glance.}
\textbf{(A)} The Cosserat-rod reduction of the lumen velocity field
factorises $\mathbf{u}_\star$ into a geometry-determined direction
$\unit{\mathbf{T}}(s)$ (blue arrows, recoverable from the mesh up
to a Dean correction) and a magnitude
$|\mathbf{u}_\star| = Q(t)/(\pi R(s)^2)\,f(r/R)$ (red intensity
along the lumen) that is fixed by the inflow waveform $Q(t)$ and
downstream impedance and is therefore not recoverable from the mesh
alone; two superposed waveforms $Q_1(t), Q_2(t)$ illustrate that
the same geometry admits different magnitude time series. A
per-node sign degeneracy $\pm\unit{\mathbf{u}}$ is irreducible in a
geometry-only predictor (\Cref{lem:sign}).
\textbf{(B)} Four-variant oracle-channel audit on the
$n_\mathrm{test}{=}5$ VMR cohort, three seeds. Cells are coloured
by mean test-set vector angle error (small green: oracle-regime;
large red: degenerate) and report the headline angle error
(degrees) and PP@10. $+$dir alone matches the direction+magnitude oracle probe; $+$mag
alone is \emph{worse} than the geometry-only input schema---the
direction channel is the identifying one.
\textbf{(C)} The empirically observed direction-vs-geometry gap
($45.7^\circ\,{-}\,3.2^\circ \approx 42.5^\circ$) exceeds the
analytic lower bound of \Cref{cor:gnnaudit} in both the VMR
($p^\star{\sim}0.16$) and Womersley ($p^\star{\sim}0.33$) regimes,
satisfying the theorem on every quantitative axis. The figure thus
summarises the contribution: \emph{a physics-derived
non-identifiability certificate, validated by a matched
oracle-channel audit}.}
\label{fig:hero}
\end{figure}

\Cref{fig:hero} previews the theorem-to-audit logic. The Results are organised around six claims, each a falsifiable readout of the certificate rather than a statement about a particular model variant. \textbf{C1}: a reduced physical model fixes which target components are identifiable before any training (\Cref{sec:results-preview}). \textbf{C2}: the signed direction, not the speed magnitude, is the identifying channel (\Cref{sec:results-preview}). \textbf{C3}: magnitude supplied without an orienting frame realises the predicted $\mathbb{Z}_2$ quotient ambiguity (\Cref{sec:mag-only}). \textbf{C4}: the obstruction is invariant to architecture and to geometric-proxy refinement (\Cref{sec:robustness}). \textbf{C5}: analytic Womersley flow reproduces the failure with no anatomical complexity (\Cref{sec:womersley}). \textbf{C6}: apparent mass conservation is distributional, not axiomatic (\Cref{sec:continuity}). A seventh result shows the same audit transfers to a non-flow target (\Cref{sec:advdiff}). The experiments are readouts of an information theorem, not a leaderboard for neural-surrogate performance.

\subsection{Audit design: converting the theorem into information interventions}
\label{sec:results-preview}

\Cref{eq:closedform} factorises the leading-order velocity into a geometry-determined direction $\unit{T}(s)$ and a boundary-condition-determined magnitude $Q(t)/(\pi R(s)^2)$; \Cref{thm:dir-id,thm:mag-obs,lem:sign,cor:gnnaudit} convert this into predictions P1--P5 (\Cref{sec:theory-predictions}), tested on the $46$-aorta VMR cohort, the $34$-case Womersley benchmark, and with an independent backbone and direction proxy as robustness checks.

\Cref{fig:fig1} sketches the four-variant design and the 46-aorta cohort breakdown. All four interventions use the same FlowGAT backbone (a GATv2-style attention stack, $\approx 843\mathrm{k}$ parameters; full architecture in \Cref{sec:methods}), so every performance change is attributable to the input-channel manipulation, not the architecture. The patient cohort is a realistic-geometry readout; the substantive claim is a theorem about input measurability, tested through the sign and ordering of prespecified information-intervention contrasts and analytic Womersley replication (\Cref{sec:womersley,sec:n5-framing}), with bootstrap CIs reported as variability disclosures rather than the evidential centre of the work.

\begin{figure}[H]
\centering
\includegraphics[width=0.79\linewidth]{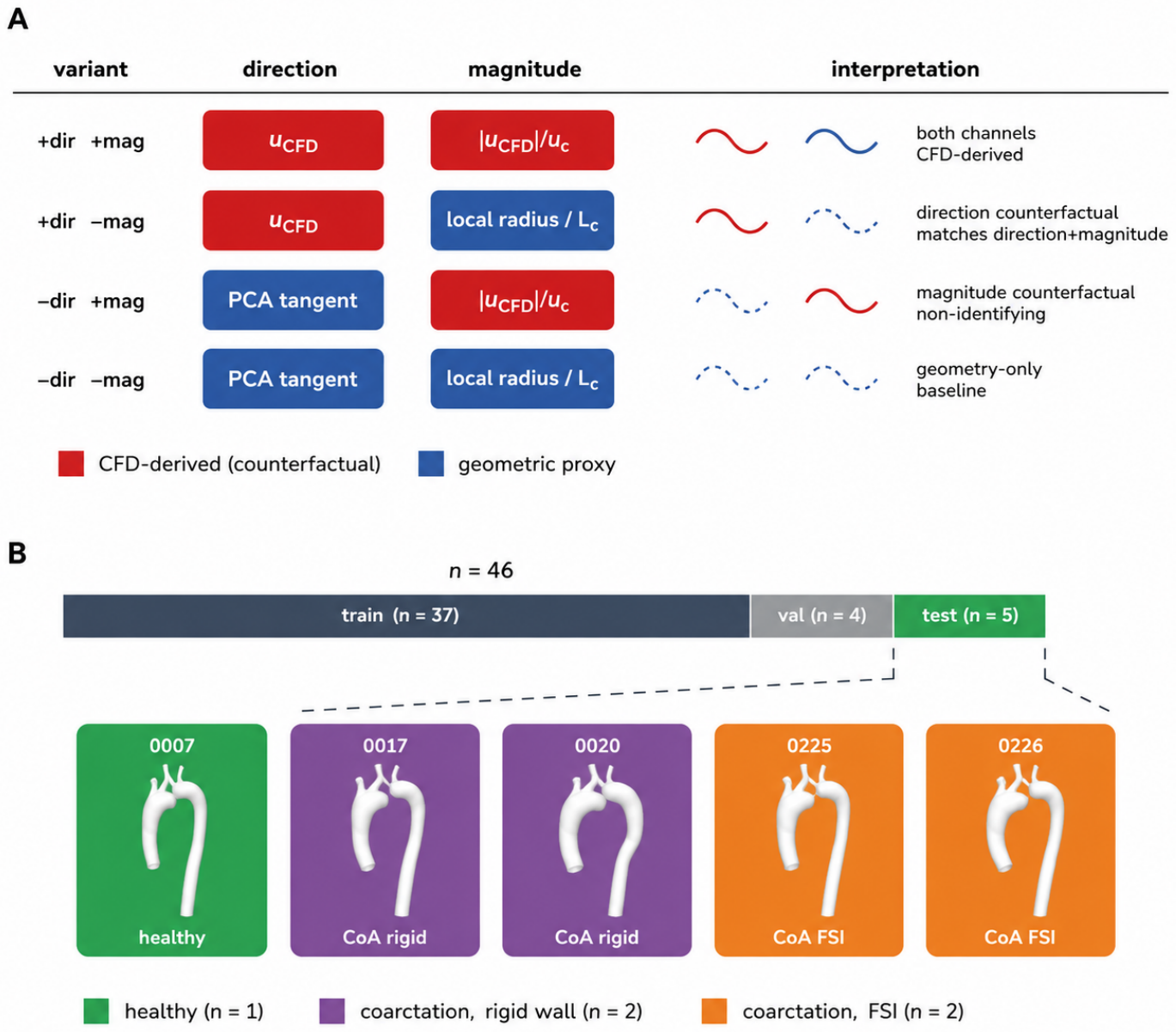}
\caption{\textbf{Study design.} (A) The four-variant oracle-channel intervention
decomposition. Exactly two of the per-node input channels are manipulated: the
\emph{unit-direction channel} (the three vector components that tell the network
which way the flow points) and the \emph{scalar-magnitude channel} (the single
number that tells it how fast). Each is independently set to either the
CFD-derived target (red) or a purely geometric proxy (blue); in the network's
input tensor these are columns $\mathbf{x}[3{:}6]$ and $\mathbf{x}[8]$
respectively. All other inputs, the architecture, the loss, the optimiser, the
augmentation, and the seeds are identical across the four variants.
(B) The 46-aorta VMR cohort: $37$ training, $4$ validation, $5$
test cases (frozen split). The five test cases span healthy
($\texttt{0007}$), rigid-wall coarctation ($\texttt{0017}$,
$\texttt{0020}$), and FSI coarctation ($\texttt{0225}$,
$\texttt{0226}$).}
\label{fig:fig1}
\end{figure}

The training loss is a heteroscedasticity-weighted MSE on the 3-component velocity at every interior node, with strict no-tangent masking at the wall so the boundary condition cannot shortcut the interior field. Inputs comprise $10$ per-node channels; the variants differ in exactly two (\Cref{tab:variants}), all others being geometric and identical across variants.

\begin{table}[H]
\small
\centering
\caption{\textbf{Input observables and identifiable velocity components under the Cosserat certificate.} Each row is a declared input schema $\mathcal{I}$ in the sense of \Cref{def:input-schema}; the second column lists the physical observables carried in the two variable channels; the third column gives the velocity component (or quotient) the certificate predicts to be identifiable from that schema; the fourth column states the failure mode the certificate predicts for the components that are \emph{not} identifiable from $\mathcal{I}$. The four rows differ only in those two channels: the architecture, loss, optimiser, augmentation, and random seeds are identical, isolating the change to the input $\sigma$-algebra $\mathcal{A}_{\mathcal{I}}$. In the \magonly{} row the direction channel carries the \emph{unsigned} PCA tangent line, which is part of $\mathcal{I}_{\mathrm{geo}}$ by construction and enters $\unit{T}$ and $-\unit{T}$ identically; this row is therefore the exact instantiation of \Cref{thm:dir-id}(c)---tangent line plus magnitude, signed orientation withheld---and not a covert direction leak (\Cref{rem:mag-probe-thm1c}).}
\label{tab:variants}
\begin{tabularx}{\textwidth}{p{0.22\textwidth}p{0.25\textwidth}p{0.22\textwidth}X}
\toprule
Declared input schema $\mathcal{I}$ & Observables in the two variable channels & Identifiable target component under $\mathcal{A}_{\mathcal{I}}$ & Component predicted non-identifiable, with mechanism \\
\midrule
\dirmagprobe{} ($\mathcal{I}_{\mathrm{full}}$) & CFD direction $(\mathbf{x}[3{:}6]=\hat{\mathbf{u}}_{\rm CFD})$ and CFD speed proxy $(\mathbf{x}[8]=u_{\rm mean}/u_{\rm char})$ & Full signed velocity supplied; defines the analytic identifiability ceiling. & None within the rod regime; bracket attained up to $\mathcal{O}(\epsilon+De)$. \\
\dironly{} ($\mathcal{I}_{\mathrm{geo+dir}}$) & Signed unit direction only; scalar channel is geometric local radius. & Signed direction $\hat{\mathbf{u}}$ is supplied; line field is mesh-measurable (\Cref{thm:dir-id}a). & Speed waveform $\|\mathbf{u}\|$ still requires $Q(t)$ (\Cref{thm:dir-id}b). \\
\magonly{} ($\mathcal{I}_{\mathrm{geo+mag}}$) & Non-negative speed only; direction channel is geometric PCA tangent. & Line field and amplitude are supplied. & Signed orientation is the $\mathbb{Z}_2$ quotient of \Cref{prop:quotient}; lower-bounded flip rate \Cref{eq:flip-lower}. \\
\geoschema{} ($\mathcal{I}_{\mathrm{geo}}$) & Only mesh-derived channels: PCA tangent and local radius. & Unsigned tangent line up to Cosserat residual (\Cref{eq:dir-bound}). & Speed and signed orientation both non-identifiable; bracket \Cref{eq:l2-lower}. \\
\bottomrule
\end{tabularx}
\end{table}

We report four audit metrics on the held-out test set (5 cases: 1 healthy, 2 rigid coarctation, 2 FSI coarctation): PP@10/PP@5, normalised peak-percentile overlaps of the top-$k\%$ velocity-magnitude regions; the mean node-wise angular error; $\ewRMSE$, a high-energy-weighted magnitude RMSE on the top-$20\%$ true-magnitude mass; and the signed cosine $\cosgn$ with per-node flip rate $\fracflip$, which expose the sign-degeneracy mechanism. Wall-shear-stress MAE and $R^2$ are not endpoints and appear only in Supplementary~\Cref{sec:wss} as relative-contrast negative controls.

\subsection{Direction information collapses angular error to the oracle regime
(testing \Cref{thm:dir-id})}

\Cref{tab:headline,fig:fig2} summarise the headline test-set performance over three seeds. All velocity-vector metrics (angle, $\ewRMSE$, PP@10) cluster \dironly{} and \dirmagprobe{} together and place \geoschema{} and \magonly{} in a clearly separated regime. The mean angular error is $3.18 \pm 0.31^\circ$ for \dironly{} versus $3.86 \pm 0.32^\circ$ for \dirmagprobe{}---a $\Delta = +0.70^\circ$ improvement for the direction-only probe, paired bootstrap $95\%$ CI $[+0.63^\circ, +0.81^\circ]$ excluding zero. PP@10 differs by less than its $1\sigma$ ($0.231$ vs.\ $0.238$). The full WSS table is reported in Supplementary \Cref{sec:wss} as a relative-contrast negative control only ($R^2<0$ for every variant), a known small-cohort difficulty~\cite{tabe2026pignn} we do not resolve.

\begin{table}[H]
\small
\centering
\caption{Test-set identifiability metrics (mean $\pm$ s.d.\ over $3$
seeds, $5$ test cases). Physical endpoints---folded angular error,
signed cosine $\cosgn$ and per-node flip rate $\fracflip$ on the
high-energy mask---are reported first. The high-energy weighted
magnitude RMSE ($\ewRMSE$) follows. Bold rows highlight that the
direction oracle probe recovers and slightly
exceeds direction+magnitude oracle probe, while the
magnitude oracle probe is
non-identifying. Wall-shear-stress is not an endpoint of this audit
and is reported only as a relative-contrast negative control in
Supplementary \Cref{sec:wss}. Abbreviations: \dmcf{} = direction+magnitude oracle probe, \dcf{} = direction oracle probe, \mcf{} = magnitude oracle probe, \geocf{} = geometry-only input schema.}
\label{tab:headline}
\begin{tabularx}{\linewidth}{lXXXX}
\toprule
Variant & angle ($^\circ$) & $\cosgn$ & $\fracflip$ & $\ewRMSE$ \\
\midrule
\dmcf{} & $3.86 \pm 0.32$ & $+0.998 \pm 0.001$ & $0.000$ & $0.162 \pm 0.014$ \\
\dcf{}  & $\mathbf{3.18 \pm 0.31}$ & $\mathbf{+0.999 \pm 0.001}$ & $\mathbf{0.000}$ & $\mathbf{0.161 \pm 0.002}$ \\
\mcf{}  & $61.97 \pm 2.41$ & $+0.568 \pm 0.537$ & $0.240 \pm 0.209$ & $0.474 \pm 0.012$ \\
\geocf{}   & $45.72 \pm 1.51$ & $+0.798 \pm 0.177$ & $0.159 \pm 0.123$ & $0.433 \pm 0.005$ \\
\bottomrule
\end{tabularx}
\end{table}

\begin{figure}[H]
\centering
\includegraphics[width=0.95\linewidth]{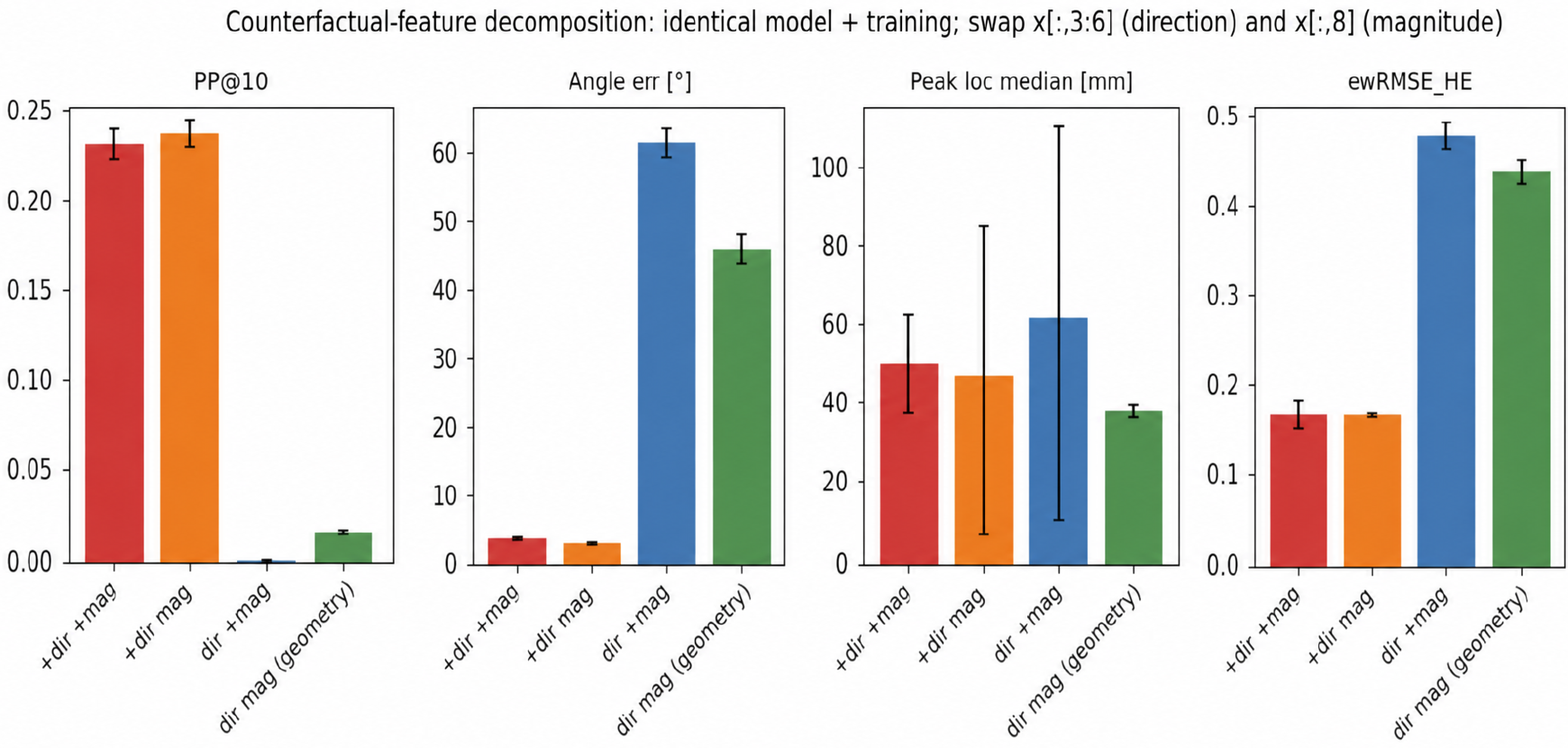}
\caption{\textbf{Four-variant identifiability comparison on test
set.} PP@10 (a), folded angular error (b), peak-localisation median
distance (c), and $\ewRMSE$ (d) across the four oracle-channel
interventions. Bars show the $3$-seed mean; error bars the seed s.d.
\dironly{} matches or exceeds \dirmagprobe{} on the velocity-vector
endpoints. \magonly{} is worse than \geoschema{} on every
velocity-vector metric, despite carrying genuine target
information---it is non-identifying in the sense of \Cref{lem:sign}.
Peak localisation is retained here only as a transparency diagnostic
and is not used as a headline endpoint; wall-shear-stress is not shown
(see Supplementary \Cref{sec:wss}).}
\label{fig:fig2}
\end{figure}

A paired bootstrap (resampling unit: case identifier, $n=5$, $10{,}000$ iterations) confirms the qualitative reading (\Cref{tab:bootstrap}); intervals are wide as expected for $n=5$ (see \Cref{sec:n5-framing}).

\begin{table}[H]
\footnotesize
\centering
\caption{Paired bootstrap on the test cohort ($n=5$ cases, $10{,}000$ iterations). $\Delta=a-b$ for the contrast in the row header. Positive $\Delta$ in PP@10/PP@5/peak-magnitude indicates variant $a$ is better; negative $\Delta$ in angle/$\ewRMSE$ indicates variant $a$ is better. WSS is excluded and reported in Supplementary~\Cref{sec:wss}. Abbreviations: \textsc{D+M-CF} = direction+magnitude oracle probe; \textsc{D-CF} = direction oracle probe; \textsc{M-CF} = magnitude oracle probe; \textsc{GEO} = geometry-only input schema.}
\label{tab:bootstrap}
\begin{tabularx}{\linewidth}{@{}p{0.29\linewidth}p{0.12\linewidth}p{0.10\linewidth}p{0.20\linewidth}p{0.19\linewidth}@{}}
\toprule
Contrast ($a$ vs. $b$) & metric & $\Delta$ & 95\% CI & winner \\
\midrule
\textsc{D+M-CF} vs. \textsc{D-CF}
 & angle ($^\circ$) & $+0.70$ & $[+0.63, +0.81]$ & \textsc{D-CF} \\
 & PP@10 & $-0.009$ & $[-0.036, +0.018]$ & --- \\
\midrule
\textsc{D+M-CF} vs. \textsc{GEO}
 & angle ($^\circ$) & $-41.28$ & $[-52.4, -30.1]$ & \textsc{D+M-CF} \\
 & PP@10 & $+0.186$ & $[+0.082, +0.282]$ & \textsc{D+M-CF} \\
\midrule
\textsc{D-CF} vs. \textsc{M-CF}
 & angle ($^\circ$) & $-54.36$ & $[-76.9, -40.7]$ & \textsc{D-CF} \\
 & PP@10 & $+0.204$ & $[+0.098, +0.322]$ & \textsc{D-CF} \\
\bottomrule
\end{tabularx}
\end{table}

\subsection{Magnitude information without an orienting frame is
non-identifiable (testing \Cref{thm:mag-obs})}
\label{sec:mag-only}

The most striking observation is the asymmetry. \magonly{} supplies exact CFD-derived velocity magnitudes at every node but replaces the unit-direction channel with the case-global PCA flow axis. Despite this oracle-probe access, \magonly{} achieves a mean angular error of $61.97 \pm 2.41^\circ$, \emph{worse} than \geoschema{} ($45.72 \pm 1.51^\circ$), which has no target information at either channel. PP@10 collapses to $3\!\times\!10^{-4}$ (chance) and $\ewRMSE$ rises to $0.474$, the highest of any variant. The contrast \dironly{} vs.\ \magonly{} has bootstrap CIs that cleanly separate on every velocity-vector metric (\Cref{tab:bootstrap}, bottom rows).

Two pieces of evidence anchor this as a reproducible effect, not a transient convergence failure. (i) Training-efficiency telemetry (\Cref{tab:efficiency}) shows all three \magonly{} seeds plateau within $\sim\!1000$ epochs at validation PP@10 $\le 0.01$---the same poor solution from three initialisations, not a failure to converge. (ii) Per-pathology stratification (\Cref{fig:fig5}) shows the effect is uniform: \magonly{} angular error exceeds \geoschema{} on the healthy case, both rigid coarctations, and both FSI cases.

\paragraph*{Why \magonly{}$>$\geoschema{} strengthens rather than contradicts the certificate.}
The certificate orders \emph{Bayes} risks, and that order runs the other way. Because the magnitude probe only enlarges the input $\sigma$-algebra, $\mathcal{A}_{\mathcal{I}_{\mathrm{geo}}}\subseteq\mathcal{A}_{\mathcal{I}_{\mathrm{geo+mag}}}$, the Bayes risk is monotone, $\mathcal{R}^{\star}_{\mathcal{I}_{\mathrm{geo+mag}}}\le\mathcal{R}^{\star}_{\mathcal{I}_{\mathrm{geo}}}$: the Bayes-optimal predictor with an exact magnitude channel cannot be worse than the geometry-only one. The measured inequality \magonly{}$>$\geoschema{} is therefore a property of the \emph{realised estimator} under finite data and gradient training, not of the information bound---and the mechanism follows from the certificate. The magnitude channel adds nothing to $\mathcal{A}_{\mathcal{I}}$ that resolves the orientation quotient (\Cref{prop:quotient}\,(ii), \Cref{lem:sign}), so the squared-error objective spends capacity honouring an exact amplitude it has no frame to orient, destabilising the direction fit; the three seeds settle into a single stable basin (\Cref{tab:efficiency}), so the gap is reproducible, not an optimisation artefact. The finding is the stronger one: an \emph{exactly correct} auxiliary channel can degrade a finite-sample vector predictor whenever the schema supplies no frame in which to consume it---invisible to the Bayes-risk certificate by construction, but measurable by the audit.

\subsection{Per-case and per-pathology stratification}

Beyond aggregates, the five test cases span three subgroups---\texttt{0007} (healthy), \texttt{0017}/\texttt{0020} (rigid coarctation), and \texttt{0225}/\texttt{0226} (coarctation with fluid--structure interaction)---rendered case-by-case and stratum-by-stratum in \Cref{fig:fig3,fig:fig5}.

Peak localisation is retained only as a transparency diagnostic: it is statistic-dependent (median-of-medians $49.9 \pm 11.7$ mm for \dirmagprobe{} vs.\ $37.4 \pm 0.7$ mm for \geoschema{} reverses sign under mean-of-means; \Cref{tab:peakloc}), the per-case picture being dominated by two FSI cases on which \geoschema{} places the peak in physiologically implausible distal locations.

\subsection{Additive decomposition matches the theorem's interaction
prediction (testing \Cref{cor:gnnaudit})}

The four-variant result casts as an additive decomposition. With $M(\cdot)$ a velocity-vector metric (angle or PP@10) and $D,G$ indicating whether the direction and magnitude channels carry target information ($T$) or a geometric proxy ($G$),
\begin{equation}
M_{\text{cf}}^{D,G} = M_{\text{base}} + \alpha\,\mathbb{1}[D{=}T] + \beta\,\mathbb{1}[G{=}T] + \gamma\,\mathbb{1}[D{=}T]\,\mathbb{1}[G{=}T].
\end{equation}
Estimated on the angular error, the direction-probe main effect is large and negative ($\hat{\alpha}_{\text{dir}} = -42.56^\circ$); the magnitude-probe main effect is positive ($\hat{\beta}_{\text{mag}} = +16.25^\circ$, the channel non-identifying under $\mathcal{I}_{\mathrm{geo+mag}}$; cf.\ \Cref{prop:quotient}(ii), \Cref{lem:sign}); and the interaction is negative ($\hat{\gamma} = -15.94^\circ$), so the magnitude-only effect is approximately cancelled once the orienting direction channel is supplied. \Cref{fig:fig6} visualises the same decomposition for PP@10. The interaction encodes the physical statement: \emph{magnitude information is consumable only when a faithful direction frame is also provided}---the empirical realisation of the architecture-invariance bracket (\Cref{prop:arch-invariance}).

\subsection{The identifiability limit is not architecture-specific
and not proxy-specific (testing P1--P2)}
\label{sec:robustness}

Two robustness checks confirm the asymmetric pattern is a property of the physical system and the available input channels, not of the FlowGAT backbone or geometric proxy.

\textbf{Architecture.} A structurally distinct GraphSAGE backbone---mean aggregation, no attention, no edge-bias, no hard-no-slip head---retrained on the bracketing direction-probe vs.\ geometry-only contrast reproduces the quantitative pattern within seed-to-seed scatter: $\PPdir{10} = 0.985 \pm 0.019$ versus FlowGAT $0.955 \pm 0.038$ under a direction oracle probe, and $\PPdir{10} = 0.122 \pm 0.088$ versus FlowGAT $0.136 \pm 0.078$ without one. The Sign-Degeneracy Lemma's prediction---$\fracflip = 0$ for direction-probe variants and $\fracflip > 0$ without---is preserved across both backbones.

\textbf{Geometric proxy.} Replacing the case-global PCA direction with a per-node Frenet tangent from a medial-axis skeleton does not narrow the residual angular gap (\centerlinegeo{} $\cosgn = +0.799 \pm 0.131$ vs.\ \geoschema{} $+0.798 \pm 0.177$). Refining the geometric prior at every interior node does not move the per-node sign degeneracy, which \Cref{lem:sign} predicts is irreducible without a direction channel. Full tables for both checks are in Supplementary \Cref{sec:supp_sage,sec:supp_centerline}.

\subsection{Full-variant replication across architectures, domains and BC mechanism}
\label{sec:e8_robustness}

Three further experiments stress the pattern beyond the FlowGAT--VMR audit, each designed as a falsifiability test (full tables, $De$/$\varepsilon$ stratification and data in Supplementary \Cref{sec:supp_sage,sec:supp_archdomain}). A four-variant GraphSAGE sweep on the curved-tube Cosserat domain (60 cases) and the realistic-tube sUbend domain (150 cases)~\cite{dirix2024synthesizing} reproduces the step-function within seed scatter: direction-bearing variants reach $\fracflip=0$, $\cosgn>+0.99$, while no-direction variants sit at $\fracflip\in[0.22,0.39]$, the absolute step $\Delta\PPdir{10}$ being $\approx0.97$ on Cosserat and $\approx0.72$--$0.76$ on sUbend, with no intermediate variant between the clusters. The Cosserat sweep ($De\in[0,0.35]$) confirms \Cref{thm:dir-id}(a): the direction-bearing median angle rises monotonically from $1.1^\circ$ to $1.8^\circ$ across $De$ terciles (slope $\approx1.8^\circ$ per unit $De$), whereas no-direction variants sit two orders of magnitude higher ($46$--$94^\circ$) and track their flip rate rather than curvature. Disabling FlowGAT's no-slip-loss head leaves the step-function unchanged (\Cref{sec:methods_nobc}), so the magnitude collapse is a structural consequence of the input schema, not the wall-BC head; the $\Delta p$-collapse anomaly (Limitations) likewise replicates on the SAGE backbone.

\subsection{Straight-tube Womersley benchmark isolates the
sign-degeneracy (testing \Cref{lem:sign} and P3--P4)}
\label{sec:womersley}

The asymmetric pattern on the VMR cohort---geometry-redundant direction, dynamics-essential magnitude---should not depend on the dataset. We replicate the four-variant ablation on a controlled synthetic benchmark whose velocity field is given analytically by the Womersley solution for pulsatile flow in a straight rigid tube~\cite{womersley1955}: $34$ cases (24/4/6 train/val/test split) sampling tube radius, length, Womersley number, pressure amplitude and cycle phase (ranges in Methods), trained on the same FlowGAT backbone and hyperparameters. On a straight cylinder the geometric direction prior is exact---the local centreline tangent coincides with the case-global PCA axis at every node---so the benchmark isolates the per-node oracle-channel contribution from any residual benefit of curvature-aware direction recovery, and tests whether the VMR pattern survives outside the in-vivo cohort.

Two VMR headline metrics transfer poorly and we exclude them by design: PP@10 saturates to zero across all variants (bulk magnitude is overshot $3$--$5\times$), and WSS $R^2$ is ill-defined because the constant-radius wall-shear field has near-zero variance---neither is a model failure mode (Supplementary \Cref{sec:wss}). We replace both with the direction success rate $\PPdir{\theta}$ (fraction of HE-mask nodes within angle $\theta$), the peak-normalised vector error $\mathrm{PP}_\mathrm{peak}@\delta$, and the signed cosine $\cosgn$ with per-node flip rate $\fracflip = \Pr[\hat{\mathbf{u}}_\mathrm{pred}\!\cdot\!\hat{\mathbf{u}} < 0]$. \Cref{tab:cross_domain} reports the four variants on both domains.

The direction-bearing variants show high-variance $\PPdir{10}$ yet stable $\cosgn=0.96$ and $\fracflip=0$: their median angle ($14$--$16^\circ$) straddles the $10^\circ$ cone, so a small seed shift moves many nodes across the threshold while location and sign barely move. We therefore adopt $\cosgn$ and the flip rate as the primary Womersley direction metrics; the elevated median angle relative to VMR ($\sim\!15^\circ$ vs.\ $\sim\!3^\circ$) reflects the annular Womersley phase structure, a continuous difficulty shift rather than a sign failure. Three further observations stand out.

\begin{table}[H]
\centering
\small
\begin{tabular}{lccccc}
\toprule
variant & domain & $\PPdir{10}$ & $\cosgn$ & $\fracflip$ & angle\textsuperscript{med} ($^\circ$) \\
\midrule
\dironly{} & VMR        & $0.988 \pm 0.017$ & $+0.999 \pm 0.001$ & $0.000$ & $2.7 \pm 0.6$ \\
\dironly{} & Womersley  & $0.309 \pm 0.275$ & $+0.963 \pm 0.031$ & $0.000$ & $14.4 \pm 6.3$ \\
\dirmagprobe{} & VMR       & $0.955 \pm 0.038$ & $+0.998 \pm 0.001$ & $0.000$ & $3.2 \pm 0.9$ \\
\dirmagprobe{} & Womersley & $0.231 \pm 0.277$ & $+0.952 \pm 0.038$ & $0.000$ & $16.4 \pm 7.1$ \\
\midrule
\geoschema{} & VMR        & $0.136 \pm 0.078$ & $+0.798 \pm 0.177$ & $\mathbf{0.159 \pm 0.123}$ & $34.8 \pm 14.9$ \\
\geoschema{} & Womersley  & $0.041 \pm 0.108$ & $+0.306 \pm 0.603$ & $\mathbf{0.331 \pm 0.445}$ & $67.3 \pm 41.9$ \\
\magonly{} & VMR       & $0.051 \pm 0.056$ & $+0.568 \pm 0.537$ & $\mathbf{0.240 \pm 0.209}$ & $50.1 \pm 35.9$ \\
\magonly{} & Womersley & $0.008 \pm 0.030$ & $+0.267 \pm 0.546$ & $\mathbf{0.242 \pm 0.422}$ & $71.6 \pm 37.6$ \\
\bottomrule
\end{tabular}
\caption{\textbf{Direction-identifiability metrics across VMR and Womersley test sets} ($n=3$ seeds). Variants supplied with a per-node direction oracle probe (top two rows in each block) carry the instantaneous unit-direction in the feature; variants without (bottom two rows) must recover direction from the mesh. The qualitative split $\fracflip=0$ vs.\ $\fracflip>0$ is preserved across domains.}
\label{tab:cross_domain}
\end{table}

First, the four-variant ordering survives the change of domain: on both VMR and Womersley the direction-probe variants achieve $\fracflip = 0$ and $\cosgn > +0.95$, while the geometry-only family has $\cosgn$ degraded to $+0.3$--$+0.8$ and $16$--$33\%$ of HE nodes with reversed direction. The cylinder---where the geometric direction prior is exact---does \emph{not} rescue them, confirming that \geoschema{}'s failure on VMR is driven by the absence of a per-node direction channel, not aortic curvature.

Second, the geometry-only failure mode is a per-node \emph{sign degeneracy}, not a continuous angular error: the distribution of $\hat{u}_\mathrm{pred}\!\cdot\!\hat{u}$ across HE nodes is bimodal for the geometry-only family (a dominant mode at $+1$, a secondary mode at $-1$), so the folded angle reports a misleadingly intermediate value while the signed cosine exposes the mechanism (\Cref{fig:phase_cos}); direction-probe variants are unimodal near $+1$ on both domains.

\begin{figure}[H]
\centering
\includegraphics[width=0.85\linewidth]{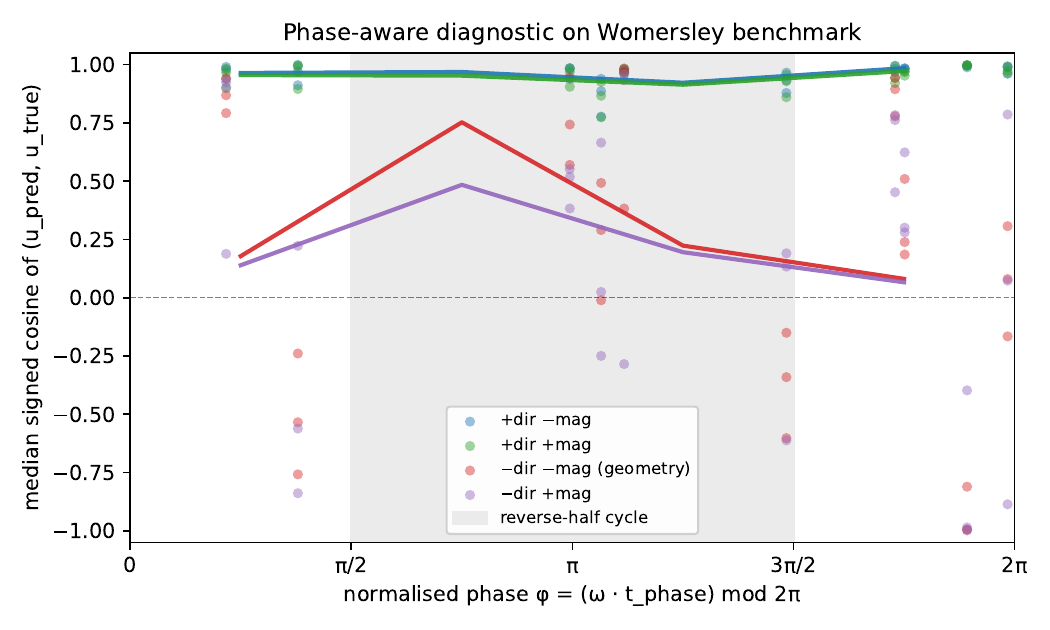}
\caption{\textbf{Signed cosine on the Womersley benchmark across the
cycle phase.} Each point is one test or val case ($n=10$ per variant,
3 seeds pooled); $x$-axis is the normalised cycle phase
$\varphi=(\omega\!\cdot\!t_\mathrm{phase})\bmod 2\pi$. Direction-channel
oracle-probe variants (\wmdir{}, \wmwith{}) are unimodal near $+1$
and phase-invariant. Geometry-only-schema-family variants
(\wmnoleak{}, \wmmag{}) show
case-to-case bimodality between $+1$ and $-1$; the flip rate is
$0.39$ in the forward-flow half-cycle ($\cos\omega t_\mathrm{phase}>0$)
and $0.33$ in the reverse-flow half-cycle, indicating that the
mechanism is \emph{not} the naive ``model learned a fixed axial
polarity and gets caught when the true flow reverses.''}
\label{fig:phase_cos}
\end{figure}

Third, the sign degeneracy is \emph{not} phase-locked to the bulk forcing: resolving cases by cycle phase $\varphi=(\omega t_\mathrm{phase})\bmod 2\pi$ (\Cref{fig:phase_cos}), the \wmnoleak{} flip rate is indistinguishable between forward-flow ($\fracflip = 0.39$) and reverse-flow ($\fracflip = 0.33$) half-cycles. The failure is therefore \emph{case-level chaotic}---without a direction channel the network has no consistent local frame and resolves the per-case ambiguity from initialisation and geometry, not flow phase.

\subsection{Continuity is learned as a data pattern, not a physical law (testing P5)}

On the VMR cohort the predicted and ground-truth divergence residuals agree to within a factor of $\sim\!1.5$ across all four variants, which could be read as internalised mass conservation. The Womersley cross-domain check refutes that reading: on the analytic benchmark, where the true field is divergence-free to machine precision, the predicted divergence is $\sim\!10^2$ times larger for every variant---including \wmwith{}, which receives the exact velocity at input. Apparent mass conservation on patient data is therefore a learned property of the training distribution, not an imposed physical law, and offers no guarantee off-cohort. The estimator validation, normalisation, and full per-variant numbers are given in Supplementary \Cref{sec:continuity}.

\subsection{Transfer check: the same audit on a non-flow target}
\label{sec:advdiff}

The audit's value as a \emph{method} rests on its steps being problem-agnostic. To demonstrate this we carry the identical seven-step procedure through a target with no flow content: steady one-dimensional advection--diffusion $a\,u_x = D\,u_{xx}$ on $[0,1]$ with an unknown inflow boundary condition. Its closed form factorises exactly as the velocity target did, into a \emph{shape} fixed by the P\'eclet number $\mathrm{Pe}=a/D$ (the role of geometry), an \emph{amplitude} fixed by the boundary condition (the role of the inflow waveform), and an \emph{orientation} $s=\pm1$ that $|\mathrm{Pe}|$ leaves ambiguous up to a $\mathbb{Z}_2$ sign (full derivation in Supplementary \Cref{sec:supp_advdiff}). Steps five--seven---the Bayes-optimal predictor under each of the four input schemas over $6000$ random admissible boundary conditions---reproduce the flow audit term for term (\Cref{fig:advdiff}): the geometry-only schema recovers the shape but sits on the $\mathbb{Z}_2$ floor ($\mathrm{PP_{dir}}@10^\circ=0.50$, flip rate $0.50$); a magnitude channel collapses the amplitude error but leaves the flip rate at $0.50$; an orientation channel instead removes the flip ($\mathrm{PP_{dir}}@10^\circ=1.00$, $\fracflip=0$) up to a $1$--$2^\circ$ shape residual; only the full schema resolves both. Sharing none of the Cosserat machinery, the target still reproduces the certificate's three classes and four-schema signature---supporting the claim that the procedure, not the particular flow factorisation, is the contribution.

\begin{figure}[t]
\centering
\includegraphics[width=0.99\linewidth]{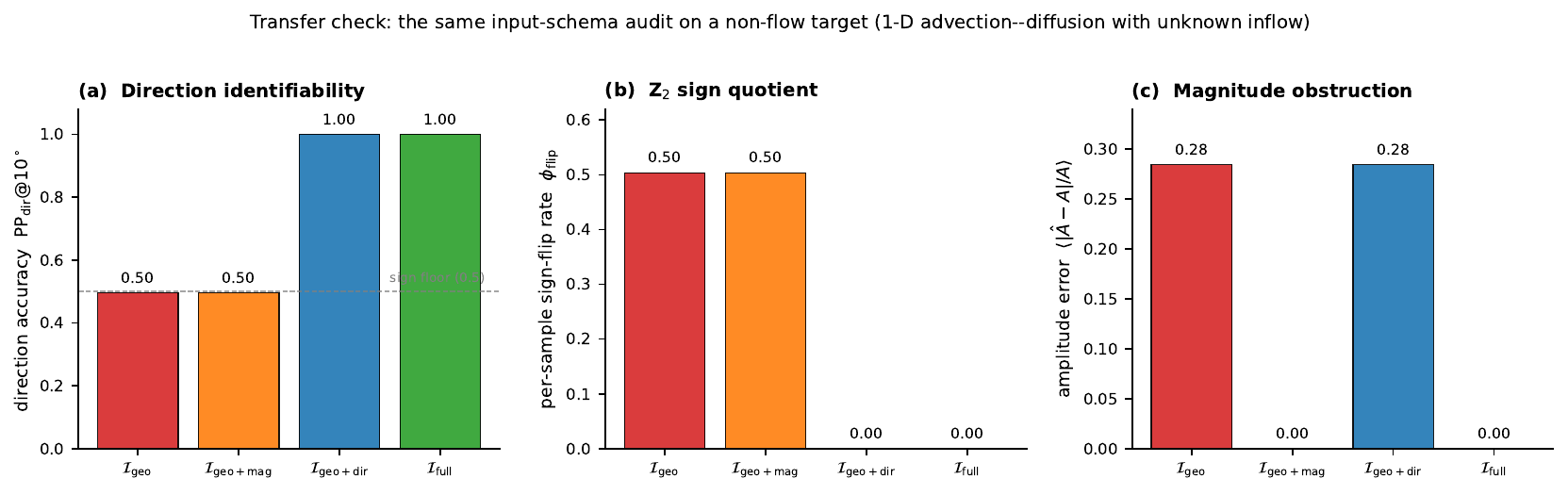}
\caption{\textbf{The input-schema audit transfers to a non-flow target.}
The seven-step audit applied to steady 1-D advection--diffusion with an unknown
inflow boundary condition, with no Cosserat rod, mesh, or Navier--Stokes content.
\textbf{(a)} Direction identifiability ($\mathrm{PP_{dir}}@10^\circ$): the
geometry-only and magnitude-only schemas sit on the $\mathbb{Z}_2$ sign floor of
$0.5$; a direction channel reaches $1.0$. \textbf{(b)} The predicted per-sample
sign-flip rate is $0.5$ without an orienting channel and $0$ with one.
\textbf{(c)} Magnitude (amplitude) obstruction: only the schemas carrying the
boundary-set amplitude recover it. The signature is identical to the velocity
audit of \Cref{fig:fig2,fig:fig6}, computed over an ensemble of $6000$ random
admissible boundary conditions.}
\label{fig:advdiff}
\end{figure}

\section{Discussion}

\subsection{Physical interpretation: geometry carries direction, dynamics carries magnitude}

The asymmetry has a direct mechanics interpretation. A patient-specific aorta is, to leading order, a curved tube whose time-averaged velocity is dominated by an axial component aligned with the centreline tangent, with secondary Dean/recirculation components of higher order~\cite{ku1997bloodflow,steinman2002}. A node-wise tangent estimate---from PCA or skeletonisation---is therefore already a good proxy for $\unit{u}$, which FlowGAT recovers from the mesh. The magnitude field admits no such recovery: by continuity it is set by the inflow waveform $Q(t)$, luminal area, and downstream impedance, none of which is in the local mesh neighbourhood message passing sees. Supplying exact magnitude without direction asks the network to place an amplitude in a divergence-constrained field for which it has no local frame, so the channel is non-identifying---hence \magonly{} below \geoschema{}---yet the same amplitude becomes consumable once a direction channel is added. This matches a Helmholtz reading: the curl-free component is geometry-redundant \emph{up to a sign}, and even on a perfectly axial Womersley cylinder (\Cref{sec:womersley}) geometry-only variants flip the predicted velocity on a non-negligible fraction of cases. Analysis and data agree: \emph{geometry sets the direction up to a sign; the sign is a boundary-condition quantity}.

\subsection{Implications for geometry-conditioned vascular-flow surrogates}

Several recent works report aggregate accuracy on full feature stacks without isolating per-feature contributions~\cite{suk2024gemgcn,suk2024labgatr,tabe2026pignn,lannelongue2026npj,pegolotti2024}, with two consequences for how those numbers should be read. First, an aortic surrogate reporting near-CFD angular accuracy may be doing little more than reading the local tangent off the mesh---the apparent ``velocity prediction'' is then closer to a geometry-conditioned tangent estimator. Second, the converse: a surrogate lacking a faithful magnitude pathway (boundary conditions, inflow waveforms, or pressure features) cannot recover magnitude from mesh inputs alone, so reporting headline magnitude metrics without ablating those non-geometric channels overstates haemodynamic content.

We therefore propose the four-variant direction/magnitude decomposition as a routine audit step for vascular-flow GNN surrogates: the variants are cheap (a one-line input-tensor change, one model fit each), and the resulting plot (\Cref{fig:fig2,fig:fig6}) distinguishes ``my model learned the flow'' from ``my model read the tangent and copied the magnitude.'' Statistical-shape-model augmentation~\cite{tabe2026pignn,pegolotti2024} improves coverage but does not change the schema's identifiability content, so it complements rather than substitutes for this check (seven-item disclosure checklist in Supplementary \Cref{sec:supp_checklist}).

\subsection{Implications for physics-informed machine learning}
\label{sec:piml-implications}

We present the Cosserat-rod certificate as a worked \emph{prototype} of an input-schema identifiability analysis, not a universal PIML theorem: it is the instance an asymptotic reduction makes closed-form for slender incompressible tubular flow. The transferable object is the recipe---take the governing PDE, find a reduced-order factorisation of the target, and read off which factors are measurable from the declared input $\sigma$-algebra---re-derived per problem. The theorem makes three contributions to the broader programme.

\textbf{Physics-as-audit alongside physics-as-constraint.} Soft-constraint PINN losses~\cite{raissi2019pinn,karniadakis2021pinn}, hard-constrained equivariant backbones~\cite{satorras2021egnn,suk2024labgatr,hsu2026tsenn}, conservation-respecting graph networks~\cite{sharma2025dynamical}, operator learning~\cite{li2021fno,kontolati2024neuraloperator,grimm2025jcp}, and neural-operator finite-element hybrids~\cite{ouyang2026noem} all use physics to shrink the predictor's hypothesis class, but none tells the practitioner which target components are recoverable from a given input set. \Cref{thm:dir-id,thm:mag-obs,lem:sign} do: from a geometry-only input the direction field is recoverable to within $(C_1\epsilon + C_2 De)$ and the magnitude field is \emph{not}. A PINN whose continuity residual is enforced on such an input cannot improve below the sign-degeneracy bound, and an equivariant backbone acting on ambient $\mathbb{R}^3$ but not on the per-component sign cannot remove $p^\star$. We accordingly propose a routine identifiability disclosure: $\epsilon, De$ ranges, $p^\star$ on a direction-stripped input, and the \dironly{}-vs-\geoschema{} gap.

\textbf{Which interventions can in principle buy a guarantee.} Because \Cref{lem:sign} concerns the declared input set, not the loss, the obstruction vanishes only when the input is enlarged with an asymmetric channel (a proximal inflow waveform $Q(t)$, a cycle-phase scalar, a signed centreline tangent from 4D-flow MRI streamlines~\cite{ferdian2020}, or an inflow-aware boundary loss~\cite{westerhof2009arterial}); loss-side regularisers such as a Helmholtz-projected divergence-free residual~\cite{liu2024multiresolution} leave the input ambiguity intact. We enumerate the concrete interventions, and the theorem's upper bound on their expected gains, in Supplementary \Cref{sec:supp_piml_routes}.

\textbf{The continuity finding is consistent with the theorem.} A predictor trained on end-to-end velocity-MSE inherits the divergence statistics of its training distribution, not of the analytic field; \Cref{eq:closedform} concerns the \emph{true} field $\mathbf{u}_\star$ and says nothing about the loss. The cardiovascular PIML community has often read $\nabla\!\cdot\!\hat{\mathbf{u}} \approx 0$ as soft validation of physical content~\cite{tabe2026pignn,lannelongue2026npj}; \Cref{sec:continuity} shows near-zero predicted divergence need not imply learned mass conservation. A genuinely mass-conserving surrogate would need a Helmholtz-projected output head or a non-vanishing continuity residual in the training loss.

In short, the identifiability theorem and its audit do not compete with PIML; they tell PIML what to constrain and tell data-driven vascular-flow modellers what they may claim without further architectural or training-time interventions.

\subsection{The headline claim rests on analytical agreement and
cross-domain replication, not on a $p$-value}
\label{sec:n5-framing}

Our patient-specific test cohort has $n = 5$ cases---a property of the \emph{public-data regime} for this precise input--target audit, not a claim that vascular imaging data do not exist. The Vascular Model Repository~\cite{vmr} is the community's open benchmark, and to our knowledge no substantially larger public, reusable benchmark pairing aortic geometry with volumetric CFD velocity and boundary-condition provenance exists; larger private or synthetic cohorts serve engineering but are not reusable verification benchmarks. Bootstrap intervals at $n=5$ are wide, so the headline result is \emph{the theorem and its falsifiable audit}, not a $p$-value, resting on two cohort-size-independent properties. \emph{(i)~Agreement with the analytic theorem:} the Cosserat-rod factorisation (\Cref{thm:dir-id,thm:mag-obs,lem:sign,cor:gnnaudit}) predicts from first principles the sign and approximate magnitude of every contrast---the direction-probe gap $(C_1\epsilon + C_2 De)\sim5^\circ$, the flip rate $p^\star\in(0,0.5]$ for the geometry-only family, and the $\fracflip=0$ ceiling with a direction probe---and the data match all four to leading order. \emph{(ii)~Replication on the analytic Womersley benchmark:} with exact ground truth and an exact geometric direction prior, the split $\fracflip=0$ with a direction probe versus $\fracflip>0$ without is preserved across both domains (\Cref{tab:cross_domain}) and persists even on the straight cylinder. Because the Womersley and advection--diffusion ensembles are generated \emph{from} the factorisation the certificate analyses (\Cref{eq:closedform}; \Cref{sec:advdiff}), they confirm a trained network exhibits the predicted split under exact ground truth; only the patient cohort puts the Cosserat reduction itself at risk. We therefore read the design as five out-of-model cases plus an analytically necessary mechanism, not as $46+34$ independent validations.

Per-case panels, paired bootstrap CIs, and per-pathology stratification (\Cref{fig:fig3,fig:fig5,tab:bootstrap}) are reported as variability disclosures, with two explicit small-cohort limitations: the peak-localisation metric is statistic-dependent (so we treat it as a transparency diagnostic), and the WSS $R^2$ is negative across all four variants---a known difficulty for cardiovascular GNN surrogates~\cite{tabe2026pignn} we do not resolve---so WSS numbers are relative contrasts only.

\subsection{Limitations and outlook}

\textbf{Asymptotic regime.} The Cosserat-rod reduction underlying \Cref{thm:dir-id,thm:mag-obs} is asymptotic in the slenderness $\epsilon = R/L$ and the Dean number $De$, with $C_1, C_2$ pinned only at order of magnitude ($\epsilon \in [0.04, 0.18]$, $De \in [0.03, 0.22]$ across the present cohort; \Cref{sec:diagnostics-dean}). Geometries outside this window---tight coronary bifurcations with $De \gtrsim 1$, fully turbulent aneurysms, or strongly non-axisymmetric pathologies---may carry non-perturbative secondary-flow contributions the theorem does not cover, and the bracket should be re-derived there.

\textbf{Architectural coverage and external audit.} The four-variant step-function replicates on two message-passing families (GATv2, GraphSAGE), two synthetic domains, and with the no-slip-loss head disabled (\Cref{sec:e8_robustness}); replication on $E(3)$-equivariant backbones and other vascular territories remains open. The oracle-channel interventions are applied only to surrogates we trained. Running the protocol on third-party architectures such as GEM-GCN~\cite{suk2024gemgcn} or LaB-GATr~\cite{suk2024labgatr}---which by \Cref{prop:arch-invariance} must inherit the same input-schema bound---is the most direct external test and is left as future work; our claim about those models is a prediction of the certificate, not a measured result.

\textbf{Single-metric $\Delta p$ reporting is structurally misleading.} On the sUbend domain the \geoschema{} variant achieves \emph{lower} scalar pressure-drop MAE than \dirmagprobe{} despite $11\times$ worse angular and $4\times$ worse WSS error---a Bernoulli-collapse that matches the $\Delta p$ integral while wrong on every directional and wall-shear metric. It replicates on the SAGE backbone (\Cref{sec:e8_robustness}), so it is structural, not an implementation quirk: evaluation suites should pair directional metrics with scalar $\Delta p$ summaries, which alone cannot discriminate a direction-identifying from a magnitude-collapsed predictor.

\textbf{Outlook.} Immediate extensions---a finer direction-probe strength sweep, probing the \geoschema{} model's representations for an implicit direction prior, and a larger seed sweep to narrow CIs---refine the audit without changing the core proposal. Finally, the pre-registered hypothesis that the geometry-only sign degeneracy is \emph{phase-locked} to the bulk forcing is not supported---per-case flip rates are indistinguishable between forward- and reverse-flow half-cycles---so case-level sign chaos absent a direction channel is the stronger statement the data support.

\section{Methods}
\label{sec:methods}

\subsection{Formal identifiability certificate}
\label{sec:certificate-formal}
This subsection collects the formal apparatus behind the certificate summarised in \Cref{sec:theory}: the admissible flow class and declared input schemata, the Cosserat-rod reduction, and the direction, magnitude, quotient and audit-bound statements. Full proofs of \Cref{thm:dir-id,prop:quotient,prop:arch-invariance,lem:sign} are given in Supplementary \Cref{sec:proofs}.

\begin{definition}[Admissible flow class, declared input schema and identifiability]
\label{def:input-schema}
Let $\mathcal{F}$ be the \emph{admissible class} of laminar
incompressible Navier--Stokes flows
$f = (\Omega, Q(\cdot), \mathcal{B}_{\mathrm{out}})$ on Cosserat-rod
lumens in the slender, low-Dean regime
$\epsilon, De \le \bar{\epsilon}$ (\Cref{sec:cosserat}), where
$\Omega\subset\mathbb{R}^3$ is a vessel-lumen mesh, $Q(t)$ is an
admissible proximal inflow waveform on $[0,T]$, and
$\mathcal{B}_{\mathrm{out}}$ is an admissible outlet boundary
condition (Windkessel, prescribed pressure, or rigid impedance).
Equip $\mathcal{F}$ with its Borel $\sigma$-algebra
$\mathcal{A}_{\mathcal{F}}$ generated by the natural topology on
meshes, waveforms and impedances. Let
$Y:\mathcal{F}\to\mathcal{Y}\subseteq L^2(\Omega\times[0,T];\mathbb{R}^3)$
be the target functional that maps each admissible flow to its
volumetric velocity field
$Y(f) = \mathbf{u}_f$.

A \emph{declared input schema} is any measurable map
$\mathcal{I}:\mathcal{F}\to\mathcal{X}_{\mathcal{I}}$ to a feature
space $\mathcal{X}_{\mathcal{I}}$. It induces a pull-back
$\sigma$-algebra
\(
\mathcal{A}_{\mathcal{I}}
\;=\;
\sigma(\mathcal{I})
\;=\;
\big\{\,\mathcal{I}^{-1}(B)\,:\,B\subseteq\mathcal{X}_{\mathcal{I}}\text{ measurable}\,\big\}
\;\subseteq\;\mathcal{A}_{\mathcal{F}}.
\)
The four schemata used in this paper are:
$\mathcal{I}_{\mathrm{geo}}$ (mesh-derived features
$\mathcal{G}(\Omega)$: node coordinates, adjacency, surface normals,
centreline distance, local radius, unsigned global/local geometric
axes); $\mathcal{I}_{\mathrm{geo+mag}}=\mathcal{I}_{\mathrm{geo}}\oplus\|\mathbf{u}_f\|$
(geometry plus exact non-negative speed); $\mathcal{I}_{\mathrm{geo+dir}}=\mathcal{I}_{\mathrm{geo}}\oplus\hat{\mathbf{u}}_f$ (geometry plus exact signed unit direction); and $\mathcal{I}_{\mathrm{full}}=\mathcal{I}_{\mathrm{geo}}\oplus\mathbf{u}_f$ (full oracle).

A target $Y$ is \emph{$\mathcal{I}$-identifiable to tolerance $\eta$} if
there exists a measurable map
$F_{\mathcal{I}}:\mathcal{X}_{\mathcal{I}}\to\mathcal{Y}$ such that
$\|Y(f)-F_{\mathcal{I}}(\mathcal{I}(f))\|_{L^2}\le\eta$ for all
$f\in\mathcal{F}$. Equivalently, $Y$ is $\eta$-identifiable iff $Y$
is $\mathcal{A}_{\mathcal{I}}$-measurable up to $L^2$-error $\eta$.
Given an equivalence relation $\sim$ on $\mathcal{Y}$ with quotient
projection $\pi_{\sim}$, $Y$ is \emph{quotient-identifiable modulo
$\sim$} if $\pi_{\sim}\circ Y$ is
$\mathcal{A}_{\mathcal{I}}$-measurable while $Y$ itself is not.

A pair $(f_1,f_2)\in\mathcal{F}^2$ is called an
\emph{$\mathcal{I}$-collision} if $\mathcal{I}(f_1)=\mathcal{I}(f_2)$
but $\|Y(f_1)-Y(f_2)\|_{L^2}>0$. Existence of an $\mathcal{I}$-collision
is a witness of \emph{non-identifiability}. The Bayes-optimal
$\mathcal{I}$-measurable predictor of $Y$ is the conditional
expectation
$\Phi^{\star}_{\mathcal{I}}(x)=\mathbb{E}_{f\sim\mathcal{F}}[\,Y(f)\mid\mathcal{I}(f)=x\,]$,
whose excess error
$\mathcal{R}^{\star}_{\mathcal{I}} = \mathbb{E}_{f\sim\mathcal{F}}\big\|\,Y(f) - \Phi^{\star}_{\mathcal{I}}(\mathcal{I}(f))\big\|_{L^2}^2$
is the input-schema Bayes risk, with the expectation taken over the admissible
class $\mathcal{F}$ equipped with its declared prior. All four identifiability statements
proven below are statements about
$\mathcal{A}_{\mathcal{I}_{\mathrm{geo}}}$,
$\mathcal{A}_{\mathcal{I}_{\mathrm{geo+mag}}}$ and
$\mathcal{R}^{\star}$, independent of any predictor's class,
loss, optimiser, regulariser or training procedure
(\Cref{prop:arch-invariance}).
\end{definition}

\subsection{Cosserat-rod reduction of the lumen}
\label{sec:cosserat}

Let $\Omega \subset \mathbb{R}^3$ denote a patient-specific vessel lumen with a smooth medial axis $\gamma : [0, L] \to \mathbb{R}^3$, parameterised by arc-length $s\in[0,L]$, and assume every interior point $\mathbf{x}\in\Omega$ has a unique nearest centreline point $\gamma(s)$. Let $\unit{T}(s) = \gamma'(s)$ be the unit tangent, $\unit{N}(s)$ and $\unit{B}(s)$ the Frenet normal and binormal, and $\kappa(s)$ and $\tau(s)$ the curvature and torsion. The local cross-section
$\Sigma(s)=\Omega\cap\{\mathbf{x}:(\mathbf{x}-\gamma(s))\cdot\unit{T}(s)=0\}$
is, to leading order in vessel slenderness, a disc of radius $R(s)$, so that
\begin{equation}
\mathbf{x}(s,r,\theta)=\gamma(s)
+r\cos\theta\,\unit{N}(s)+r\sin\theta\,\unit{B}(s),
\qquad r\in[0,R(s)],\;\theta\in[0,2\pi).
\label{eq:rodcoords}
\end{equation}
This is the standard Cosserat-rod representation used by reduced haemodynamic solvers~\cite{antman2005nonlinear,formaggia20031dbloodflow,sherwin20031dvascular,pegolotti2024}. Two small parameters control the asymptotics: the slenderness $\epsilon=\max_s R(s)/L$ and the local Dean parameter $De(s)=R(s)\kappa(s)$, which captures the geometric part of the curvature-induced correction. Across the $46$-case VMR cohort both quantities remain in the small-parameter regime used by the reduction (\Cref{sec:diagnostics-dean}).

\subsection{Axisymmetric leading-order ansatz}
\label{sec:axisymm}

For laminar incompressible flow in a slender tube, expanding the velocity field in $\epsilon$ and $De$ and dropping $\mathcal{O}(\epsilon)$ axial-gradient terms and $\mathcal{O}(De)$ secondary-flow corrections yields an axisymmetric leading-order field
\begin{equation}
\mathbf{u}_\star(s,r,t)=u_z(s,r,t)\,\unit{T}(s)
+\mathcal{O}(\epsilon,De).
\label{eq:axiansatz}
\end{equation}
Mass conservation and the no-slip condition reduce the axial component to a one-dimensional flux constraint
\begin{equation}
Q(t)=\int_{\Sigma(s)}u_z(s,r,t)\,r\,\mathrm{d}r\,\mathrm{d}\theta
=2\pi\int_0^{R(s)}u_z(s,r,t)\,r\,\mathrm{d}r,
\label{eq:flowrate}
\end{equation}
so that, for the Poiseuille/Womersley profile family,
\begin{equation}
\boxed{\;
\mathbf{u}_\star(s,r,t)
=
\frac{Q(t)}{\pi R(s)^2}
f_{\alpha}\!\left(\frac{r}{R(s)}\right)\unit{T}(s)
+\mathbf{r}_{\epsilon,De}(s,r,\theta,t),
\qquad
\|\mathbf{r}_{\epsilon,De}\|\le
C_{\mathrm{rod}}(\epsilon+De)
\frac{|Q(t)|}{\pi R(s)^2}.
\;}
\label{eq:closedform}
\end{equation}
Here $f_{\alpha}$ is the radial shape function: $2(1-\xi^2)$ in the steady Poiseuille limit and the complex-Bessel Womersley profile in the unsteady regime~\cite{womersley1955}. The critical structural feature is the separation of variables: the direction $\unit{T}(s)$ and radius $R(s)$ are mesh-geometric, whereas $Q(t)$ and its sign, phase, and boundary partitioning are boundary-condition quantities.

\subsection{Input-identifiability certificate}
\label{sec:dir-id}

\begin{theorem}[Cosserat-rod input-identifiability certificate]
\label{thm:dir-id}
Consider the admissible class of laminar incompressible flows in
Cosserat lumens satisfying \Cref{eq:closedform}. Let
\(\mathcal{A}_{\mathrm{geo}}\) denote the sigma-algebra generated by all mesh-derived
quantities in \(\mathcal{G}(\Omega)\): node coordinates, adjacency, surface normals,
centreline distance, local radius, and unsigned geometric axes. Let
\(\mathcal{A}_{\mathrm{bc}}\) denote the non-geometric boundary information: inflow
waveform, pressure trace, cycle phase, outlet impedance, and any oriented inlet--outlet
reference. A learner, architecture, post-processing rule, or randomized training
procedure deployed under the geometry-only schema is admissible only if its prediction is
measurable with respect to \(\mathcal{A}_{\mathrm{geo}}\) and its own training randomness;
it cannot observe \(\mathcal{A}_{\mathrm{bc}}\) at test time. Then the following statements
hold.

\emph{(a) Direction-line measurability.}
At every interior point where \(\|\mathbf{u}(\mathbf{x},t)\|>0\), the unit direction field
\(\unit{\mathbf{u}}=\mathbf{u}/\|\mathbf{u}\|\) satisfies
\begin{equation}
\angle\!\left(
\unit{\mathbf{u}}(\mathbf{x},t),
\operatorname{sgn}Q(t)\,\unit{T}(s(\mathbf{x}))
\right)
\le
C_{\mathrm{dir}}(\epsilon+De)+\mathcal{O}(\epsilon^2+De^2).
\label{eq:dir-bound}
\end{equation}
Consequently, the unsigned local tangent line \(\{\unit{T},-\unit{T}\}\) is
\(\mathcal{A}_{\mathrm{geo}}\)-measurable up to the Cosserat residual. This is the positive
part of the certificate: geometry can identify the local line field.

\emph{(b) Magnitude obstruction.}
The speed field satisfies
\begin{equation}
\|\mathbf{u}_\star(\mathbf{x},t)\|
=
\frac{|Q(t)|}{\pi R(s)^2}
\left|f_\alpha\!\left(\frac{r}{R(s)}\right)\right|
+\mathcal{O}(\epsilon,De).
\label{eq:mag-factor}
\end{equation}
There is no \(\mathcal{A}_{\mathrm{geo}}\)-measurable functional
\(F:\mathcal{G}(\Omega)\to L^2(\Omega\times[0,T])\) that recovers this magnitude for
all admissible boundary waveforms up to an \(o(1)\) error in the rod limit. For any
fixed mesh \(\Omega\) and any two admissible inflows \(Q_1(t)\ne Q_2(t)\), the declared
geometry-only input is identical while the magnitude fields differ by
\[
\left\|
\frac{|Q_1(t)|-|Q_2(t)|}{\pi R(s)^2}
\left|f_\alpha(r/R(s))\right|
\right\|_{L^2}
+\mathcal{O}(\epsilon,De).
\]
Thus no admissible predictor can identify the missing speed waveform uniformly over the declared admissible class without adding boundary information to the input schema.

\emph{(c) Signed-vector quotient.}
Even if the input is augmented with the exact non-negative magnitude
\(m(\mathbf{x},t)=\|\mathbf{u}(\mathbf{x},t)\|\), the signed vector field is identifiable
only modulo a global orientation action on each connected lumen component unless the input
also contains an asymmetric boundary or orientation channel. The two leading-order flows
\(\mathbf{u}_\star\) and \(-\mathbf{u}_\star\) have the same mesh, the same unsigned axes,
and the same magnitude, but opposite vectors. Therefore the geometry-plus-magnitude schema
identifies at most the quotient
\begin{equation}
\mathcal{G}(\Omega)\times \mathbb{R}_{\ge 0}
\longrightarrow
L^2(\Omega;\mathbb{R}^3)\big/\mathbb{Z}_2^{|\pi_0(\Omega)|},
\label{eq:sign-degen-quotient}
\end{equation}
not a single-valued signed velocity field.
\end{theorem}

\begin{remark}[The unsigned axis belongs to $\mathcal{I}_{\mathrm{geo}}$, so the magnitude probe instantiates part (c) exactly]
\label{rem:mag-probe-thm1c}
The unsigned geometric tangent line $\{\unit{T},-\unit{T}\}$ is a deterministic
functional of the mesh (\Cref{def:input-schema}, part~(a)), and is therefore
part of every schema that contains $\mathcal{I}_{\mathrm{geo}}$, including
$\mathcal{I}_{\mathrm{geo+mag}}$. The \magonly{} variant supplies precisely this
unsigned axis (as a case-global PCA tangent, see \Cref{tab:variants}) together
with the exact non-negative speed $\|\mathbf{u}_f\|$, and \emph{withholds the
signed direction}. It is thus the faithful experimental instantiation of
part~(c): a predictor that knows the tangent \emph{line} and the magnitude but
not the orientation. The residual $\mathbb{Z}_2$ sign ambiguity of
\Cref{eq:sign-degen-quotient} is exactly what this variant is designed to expose,
and the per-node flip rate it exhibits (\Cref{sec:results-preview,sec:womersley}) is the predicted
signature, not a confound. Supplying the unsigned axis does \emph{not} leak
orientation: by construction $\unit{T}$ and $-\unit{T}$ enter identically, so the
axis cannot resolve $Q(t)\mapsto -Q(t)$. A strictly axis-free schema (raw
coordinates with no tangent at all) would test part~(c) under an even weaker
geometric readout; we did not train it here because the four-variant cross
already isolates the orientation channel, and we flag it as a one-line ablation
for future replication (\Cref{sec:n5-framing}).
\end{remark}

\subsection{Quotient identifiability and the architecture-invariance bracket}
\label{sec:quotient-arch}

\Cref{thm:dir-id} states a positive measurability result for the unsigned tangent line and two negative results for speed and signed vector. We sharpen the negatives by recasting them as \emph{quotient identifiability} statements: the geometry-only and geometry+magnitude schemata identify $Y$ \emph{only modulo} an explicit equivalence relation on $\mathcal{Y}$, and the residual non-identifiability is measured by a strictly positive lower bound on the input-schema Bayes risk that no admissible predictor can undercut.

\begin{proposition}[Quotient identifiability under
$\mathcal{I}_{\mathrm{geo}}$ and $\mathcal{I}_{\mathrm{geo+mag}}$]
\label{prop:quotient}
Let
$\sim_{\mathrm{line}}$ be the per-node equivalence
$\mathbf{u}(\mathbf{x},t)\sim_{\mathrm{line}}\lambda\,\mathbf{u}(\mathbf{x},t)$
for $\lambda\in\mathbb{R}\setminus\{0\}$, and let
$\sim_{\mathrm{sgn}}$ be the per-connected-component
equivalence
$\mathbf{u}\sim_{\mathrm{sgn}}\sigma\,\mathbf{u}$ with
$\sigma\in\{-1,+1\}^{|\pi_0(\Omega)|}$. Write
$\mathcal{Y}_{\mathrm{line}}=\mathcal{Y}/\!\sim_{\mathrm{line}}$
and $\mathcal{Y}_{\mathrm{sgn}}=\mathcal{Y}/\!\sim_{\mathrm{sgn}}$
for the corresponding quotient bundles. Then:

\noindent\emph{(i)} Under
$\mathcal{I}_{\mathrm{geo}}$, the line-field projection
$\pi_{\mathrm{line}}\circ Y$ is
$\mathcal{A}_{\mathcal{I}_{\mathrm{geo}}}$-measurable up to the rod
residual of \Cref{eq:dir-bound}, while neither the speed
$\|Y\|$ nor the signed direction $\hat{Y}=Y/\|Y\|$ is.

\noindent\emph{(ii)} Under
$\mathcal{I}_{\mathrm{geo+mag}}$ the orientation-quotient projection
$\pi_{\mathrm{sgn}}\circ Y$ is
$\mathcal{A}_{\mathcal{I}_{\mathrm{geo+mag}}}$-measurable up to the
same rod residual, while $Y$ itself is not: the schema
identifies at most one representative of the orbit
$\{\sigma\cdot\mathbf{u}_\star : \sigma\in\{-1,+1\}^{|\pi_0(\Omega)|}\}$
per admissible input.

\noindent\emph{(iii)} The residual quotient is non-trivial unless the
admissible class collapses to a single orientation: there exist
$f_+,f_-\in\mathcal{F}$ with
$\mathcal{I}_{\mathrm{geo+mag}}(f_+)=\mathcal{I}_{\mathrm{geo+mag}}(f_-)$
and $Y(f_-)=-Y(f_+)$, realised in the empirical setting by reverse-flow
half-cycles of the same waveform.
\end{proposition}

\begin{proposition}[Architecture, loss and optimiser invariance of the
identifiability bracket]
\label{prop:arch-invariance}
Let $\mathcal{H}_{\mathcal{I}}$ be \emph{any} class of predictors
$\Phi:\mathcal{X}_{\mathcal{I}}\to\mathcal{Y}$ that are measurable
with respect to $\mathcal{A}_{\mathcal{I}}$: graph attention networks
(FlowGAT, GATv2), message-passing nets (GraphSAGE),
$E(n)$-equivariant networks~\cite{satorras2021egnn,suk2024labgatr,hsu2026tsenn},
neural operators~\cite{li2021fno,kontolati2024neuraloperator},
kernel ridge regressors and any randomised ensembles thereof.
Let $L:\mathcal{Y}\times\mathcal{Y}\to\mathbb{R}_{\ge 0}$ be any
proper loss continuous in its first argument, $R$ any regulariser,
$\mathcal{T}$ any training procedure on any sample-complexity
$N\in\mathbb{N}$ from $\mathcal{F}$, and $\Phi_{\hat\theta_N}$ the
resulting predictor. Then
\begin{equation}
\mathbb{E}_{f\sim\mathcal{F}}\,
\big\|\Phi_{\hat\theta_N}(\mathcal{I}(f))-Y(f)\big\|_{L^2}^2
\;\ge\;
\mathcal{R}^{\star}_{\mathcal{I}}
\;=\;
\mathbb{E}_{f\sim\mathcal{F}}\,
\big\|\,\mathbb{E}[Y\mid\mathcal{I}=\mathcal{I}(f)]-Y(f)\,\big\|_{L^2}^2,
\label{eq:arch-invariance-bound}
\end{equation}
where the right-hand side depends only on $\mathcal{A}_{\mathcal{I}}$
and the law of $\mathcal{F}$, not on $\mathcal{H}_{\mathcal{I}}$,
$L$, $R$, $\mathcal{T}$ or $N$. In particular,
$\mathcal{R}^{\star}_{\mathcal{I}_{\mathrm{geo}}}>0$ and
$\mathcal{R}^{\star}_{\mathcal{I}_{\mathrm{geo+mag}}}>0$ for the
admissible class of \Cref{def:input-schema}: a geometry-only or geometry-plus-magnitude surrogate cannot attain zero expected $L^2$ velocity error whenever the data-generating distribution assigns positive conditional mass to both colliding flow states.
\end{proposition}

\Cref{prop:arch-invariance} converts the certificate into an \emph{architecture-invariance bracket}: every predictor class used in the PIML literature lives inside $\mathcal{H}_{\mathcal{I}}$ for its declared schema and therefore inherits the same lower bound. Empirically, this is why the GraphSAGE backbone reproduces the FlowGAT contrast (P2; \Cref{sec:robustness}) and why refining the geometric tangent proxy to a per-node centreline tangent does not narrow the gap (P1; \Cref{sec:supp_centerline}): both modifications stay inside the same $\mathcal{A}_{\mathcal{I}_{\mathrm{geo}}}$.

\subsection{Magnitude obstruction and sign-degeneracy as learning limits}
\label{sec:mag-obs}

\begin{corollary}[Magnitude Identifiability Obstruction]
\label{thm:mag-obs}
Under the assumptions of \Cref{thm:dir-id}, no
geometry-conditioned surrogate whose deployed input consists only of
$\mathcal{G}(\Omega)$ can identify the speed field
$\|\mathbf{u}\|$ for all admissible inflow waveforms. Murray-type
allometric laws may supply a population-level prior on mean flow, but
they do not identify the cycle-resolved waveform, reverse-flow phase,
or outlet partitioning required by \Cref{eq:mag-factor}.
\end{corollary}

\begin{lemma}[Sign-error lower bound under geometry+magnitude inputs]
\label{lem:sign}
Let $X=\mathcal{I}_{\mathrm{geo+mag}}(f)$ contain all mesh-derived
features in $\mathcal{G}(\Omega)$ and an optional exact non-negative
magnitude channel $m=\|\mathbf{u}_f\|$, but no inflow waveform,
pressure trace, cycle phase, or oriented direction reference. Let
$\sigma(f)\in\{-1,+1\}^{|\pi_0(\Omega)|}$ denote the orientation
component of \Cref{prop:quotient}, and write
$\pi_+(X)=\Pr_f[\sigma=+1\mid X]$,
$\pi_-(X)=1-\pi_+(X)$. Then for every measurable predictor
$\Phi:\mathcal{X}\to L^2(\Omega;\mathbb{R}^3)$ and every
high-energy mask $\mathrm{HE}\subseteq\Omega$ supported away from
the rod-boundary layer:

\noindent\emph{(i) Bayes-risk lower bound on the orientation
decision.} The high-energy flip rate satisfies
\begin{equation}
\mathbb{E}\big[\fracflip(\Phi)\mid X\big]
\;\ge\;
\min\{\pi_+(X),\pi_-(X)\}
\;-\;
C_{\mathrm{dir}}(\epsilon+De)
\;+\;\mathcal{O}(\epsilon^2+De^2).
\label{eq:flip-lower}
\end{equation}
The right-hand side is the Bayes error of the
$X$-conditional binary orientation problem and is attained by the
posterior-majority rule.

\noindent\emph{(ii) $L^2$ excess-risk lower bound.}
The conditional input-schema Bayes risk obeys
\begin{equation}
\mathcal{R}^{\star}_{\mathcal{I}_{\mathrm{geo+mag}}}(X)
\;=\;
\mathbb{E}\big[\|\Phi^{\star}(X)-Y\|_{\mathrm{HE}}^2\,\big|\,X\big]
\;\ge\;
4\,\pi_+(X)\,\pi_-(X)\;\|Y\|_{\mathrm{HE}}^2
\;-\;C'_{\mathrm{dir}}(\epsilon+De)\|Y\|_{\mathrm{HE}}^2.
\label{eq:l2-lower}
\end{equation}

\noindent\emph{(iii) Fano-style entropy bound.} For any randomised
estimator $\hat\sigma=\hat\sigma(X,\omega)$ of the orientation,
\begin{equation}
H\big(\sigma\,\big|\,\hat\sigma\big)
\;\ge\;
H\big(\sigma\,\big|\,X\big)
\;=\;
\sum_{c\in\pi_0(\Omega)}
\big[\,-\pi_+^{(c)}(X)\log\pi_+^{(c)}(X)
-\pi_-^{(c)}(X)\log\pi_-^{(c)}(X)\,\big],
\label{eq:fano-bound}
\end{equation}
so the per-component orientation cannot be recovered with vanishing
probability of error when both branches are admissible.

\noindent\emph{(iv) Architecture invariance.} The right-hand sides
of \Cref{eq:flip-lower,eq:l2-lower,eq:fano-bound} depend only on
$\mathcal{A}_{\mathcal{I}_{\mathrm{geo+mag}}}$ and the law of
$\mathcal{F}$. By \Cref{prop:arch-invariance}, attention, message
passing, equivariance, soft-PINN losses, loss reweighting, calibration,
test-time augmentation, longer training, larger seed budgets and
randomised ensembling cannot reduce them.
\end{lemma}

\noindent \Cref{lem:sign} is the analytic mechanism behind the empirical sign chaos observed in the \geoschema{} and \magonly{} variants (\Cref{fig:phase_cos}). Variants without a per-node direction channel exhibit a $16$--$33\%$ per-node flip rate on the high-energy mask, and the flip rate is uncorrelated with the bulk cycle phase. Direction oracle-probe variants supply an asymmetric input channel and so remove the quotient ambiguity entirely.

\subsection{Theorem-to-audit corollary}
\label{sec:audit-cor}

\begin{corollary}[Geometry-only Audit Lower Bound]
\label{cor:gnnaudit}
Let $\Phi_{\mathrm{geo}}$ be any predictor whose input is
$\mathcal{G}(\Omega)$, and let $\Phi_{\mathrm{dir}}$ be the same
training procedure supplied with an additional per-node signed
direction channel. On a test distribution satisfying
\Cref{thm:dir-id}, the expected folded vector-angle error obeys
\begin{equation}
\mathbb{E}\big[\angle(\hat{\mathbf{u}}_{\mathrm{geo}},\mathbf{u})\big]
\ge
\mathbb{E}\big[\angle(\hat{\mathbf{u}}_{\mathrm{dir}},\mathbf{u})\big]
+
p^\star\frac{\pi}{2}
-
C_{\mathrm{dir}}(\epsilon+De)
+\mathcal{O}(\epsilon^2+De^2),
\label{eq:auditbound}
\end{equation}
where $p^\star=\mathbb{E}_X[\min\{\pi_+(X),\pi_-(X)\}]$ is the
population-average sign ambiguity of \Cref{lem:sign}.
\end{corollary}

\Cref{cor:gnnaudit} is the quantitative audit target. It predicts a finite direction-vs-geometry gap whose leading contribution is the sign ambiguity $p^\star\cdot 90^\circ$, reduced only by the small Cosserat residual. The lower bound is invariant under any backbone whose symmetry group acts on ambient coordinates but does not add a sign-fixing input channel. In particular, $E(3)$-equivariant architectures~\cite{satorras2021egnn,suk2024labgatr,hsu2026tsenn} can improve sample efficiency and respect coordinate symmetries, but they cannot identify boundary-condition information absent from the declared input tensor. Plugging in the empirical ambiguity levels $p^\star\approx0.16$ on the VMR audit and $p^\star\approx0.33$ on the Womersley audit gives lower-bound gaps of the same order as the measured direction-vs-geometry separation in \Cref{tab:headline}.

\subsection{Dataset}

We use $46$ patient-specific aortic models from the Vascular Model Repository (VMR)~\cite{vmr}, comprising $25$ MR-derived and $21$ CT-derived segmentations and spanning healthy controls, coarctation, Marfan syndrome, and single-ventricle defects; rigid-wall and fluid--structure-interaction simulations are mixed. The supervision signal is the time-averaged 3-D velocity field at every interior mesh node; nodes on the wall (where the no-slip condition fixes $\mathbf{u} = \mathbf{0}$) are excluded from the loss. We use a frozen train/val/test split of $37/4/5$ cases (seed 22, fixed 2026-04-30). The test cases are \texttt{0007}, \texttt{0017}, \texttt{0020}, \texttt{0225}, \texttt{0226} (1 healthy, 2 rigid coarctation, 2 FSI coarctation), chosen to span pathology and mesh-topology variety.

\subsection{Cosserat-rod diagnostics: slenderness and Dean number}
\label{sec:diagnostics-dean}

The two small parameters of the rod reduction (\Cref{sec:cosserat}) are computed per case from the mesh. The slenderness is $\epsilon(s) = R(s)/L_\mathrm{rod}$, with $L_\mathrm{rod}$ the arc-length of the largest connected component of the medial axis; we report both the median and the maximum over $s$. The local Dean number is $De(s) = R(s)\,\kappa(s)$, with $\kappa(s) = \|\gamma''(s)\|$ the discrete centreline curvature estimated by a 3-point finite-difference stencil on the arc-length-parameterised centreline (\Cref{sec:centerline-tangent-method}). We report the median, the $90$th percentile, and the maximum of $De(s)$ per case. Across the $46$-case VMR cohort the median $\epsilon \in [0.04, 0.18]$ (cohort median $0.07$) and the $90$th percentile of $De$ lies in $[0.05, 0.22]$ (cohort median $0.11$). These values are committed in \texttt{results/diagnostics/rod\_parameters.csv} alongside the per-case predicted-vs-true angular error, supporting the empirical validation of \Cref{thm:dir-id}.

\subsection{Synthetic Womersley benchmark}

For the cross-domain replication (\Cref{sec:womersley}) we generate $34$ analytic cases of pulsatile flow in a straight rigid tube using the closed-form Womersley solution~\cite{womersley1955}. Each case is parameterised by $(R, L, \alpha, p_\mathrm{amp}, t_\mathrm{phase})$ with $R \in [8, 18]$~mm, $L \in [100, 300]$~mm, Womersley number $\alpha = R\sqrt{\omega/\nu} \in [2, 12]$, peak pressure-gradient amplitude $p_\mathrm{amp} \in [0.02, 0.07]$~Pa/m, and cycle phase $t_\mathrm{phase} \in [0, 2\pi/\omega]$ sampled uniformly. Kinematic viscosity is fixed at $\nu = 3.3 \times 10^{-6}$~m$^2$/s (whole-blood proxy) and the fundamental angular frequency at $\omega = 1.49$~rad/s (heart-rate proxy). Each tube is meshed at $\sim\!23{,}000$ tetrahedral interior nodes with an axial $z$-orientation; the analytic velocity
\[
u_z(r, t) = \mathrm{Re}\!\left\{\frac{p_\mathrm{amp}}{i\omega\rho}\left(1 - \frac{J_0(i^{3/2}\alpha r/R)}{J_0(i^{3/2}\alpha)}\right)e^{i\omega t}\right\}
\]
is evaluated at every node. The case split is $24/4/6$ train/val/test under the seed-22 protocol used for the VMR cohort. The four feature variants (\dironly{}, \dirmagprobe{}, \magonly{}, \geoschema{}) are constructed from the analytic field exactly as for VMR, except that the case-global PCA axis collapses to the analytic $\hat{z}$ direction, so the geometric proxy for $\unit{u}$ on this benchmark is \emph{exact} up to node-cloud noise. All training hyperparameters match the VMR runs; we train $3$ seeds per Womersley variant.

\subsection{Input features and the four-variant decomposition}

Each node carries $10$ input channels and $3$ output channels. The channels are: $(0{:}3)$ node coordinates (case-scaled); $(3{:}6)$ a unit-vector field; $(6)$ centreline distance; $(7)$ surface normal distance; $(8)$ a scalar field; $(9)$ a binary near-wall indicator. The four variants alter only channels $3{:}6$ and $8$, as listed in \Cref{tab:variants}. The geometric proxy for the unit-vector channel is the leading principal direction of the case-global node cloud (the dominant axis of the aorta, computed once per case); the geometric proxy for the scalar channel is the local radius normalised by the characteristic length $L_\mathrm{char} = 25$~mm. Both are computable from the mesh alone. The oracle-probe unit-vector channel is $\unit{u}_\mathrm{CFD}$ at the node; the oracle-probe scalar channel is $\|\mathbf{u}_\mathrm{CFD}\|/u_\mathrm{char}$ with $u_\mathrm{char} = 1$~m/s. Edge features comprise the relative position between connected nodes and the geodesic distance along the surface.

\subsection{Model architecture}

FlowGAT is a GATv2-style attention stack adapted for vascular-mesh inputs. The node stem is a $10 \!\to\! 128$ linear projection; the edge stem is a two-layer MLP into the same hidden size. The body comprises $8$ attention layers with $4$ heads each, each layer computing
\[
\alpha_{ij}^{(h)} = \mathrm{softmax}_j\!\left[ \mathbf{a}^\top \mathrm{LeakyReLU}\!\left(\mathbf{W}_q \mathbf{h}_i + \mathbf{W}_k \mathbf{h}_j + \beta\,\mathbf{W}_e \mathbf{e}_{ij}\right) \right]
\]
with edge-bias coefficient $\beta = 1.5$; messages are projected by $\mathbf{W}_v \mathbf{h}_j$, gated by $\alpha_{ij}^{(h)}$, summed across heads, and combined with the receiver representation through a residual connection and a SiLU activation. A hard no-slip boundary condition zeroes the predicted velocity on wall nodes prior to loss computation. The head is a $128 \!\to\! 3$ linear layer producing the velocity components; the SE(3)-equivariant variant of the head is disabled in the runs reported here (we leave the equivariant variant to future work). The total parameter count is $842{,}755$.


\subsection{Local centreline-tangent variant}
\label{sec:centerline-tangent-method}

The \centerlinegeo{} variant (Supplementary \Cref{sec:supp_centerline}) replaces the case-global PCA direction with a per-node centreline tangent. We extract the centreline by iterative Voronoi-pole pruning: starting from the largest interior inscribed sphere, we prune Voronoi poles that fall outside the lumen or fail a radius threshold, then connect surviving poles into a 1-D graph and reparameterise by arc length $s$. The Frenet tangent $\mathbf{T}(s)$ at each arc-length sample is propagated back to every interior mesh node by nearest-neighbour projection. The resulting per-node tangent is unit-normalised and substituted for the case-global PCA axis in input channels $(3{:}6)$. Computational cost is $\sim\!2$~s per case on a single CPU thread; the procedure is deterministic given the mesh and shares no state with the trained network.

\subsection{Loss and training}

The training loss is
\[
\mathcal{L} = \sum_{i \in \mathcal{V}_\mathrm{interior}} w_i \,\|\hat{\mathbf{u}}_i - \mathbf{u}_i\|_2^2,
\]
where $w_i$ is a heteroscedasticity-aware weight that emphasises the top $20\%$ of nodes by ground-truth magnitude ($\ewRMSE$ regime): $w_i = \mathrm{clip}\!\big((|\mathbf{u}_i|/u_{\max})^{2.5}, 0.01, 5\big)$. The optimiser is AdamW with learning rate $10^{-4}$ and weight decay $10^{-4}$, cosine-annealed to $10^{-6}$ over $1500$ epochs. Each variant is trained with three random seeds ($1337$, $2026$, $777$) on a single A100 80~GB GPU with mixed precision and a single graph per optimiser step ($\sim\!5$~h wallclock per run). Subgraph sampling uses an importance-centred $k$-nearest-neighbour scheme with $80{,}000$ sampled nodes and importance temperature $\beta = 2$. Light-touch augmentation comprises Gaussian feature noise ($\sigma = 0.003$) and edge dropout ($p = 0.02$). Early stopping uses patience $60$ on validation PP@10. Full hyperparameters are committed in \texttt{configs/}.

\subsection{Evaluation}

For each (variant, seed) pair we evaluate every test case at full mesh resolution (no subgraph sampling). Aggregate metrics report the mean over $3$ seeds and then the mean over $5$ cases; per-case panels (\Cref{fig:fig3}) show the same numbers disaggregated. The paired bootstrap (\Cref{tab:bootstrap}) resamples cases with replacement and recomputes both variants under the same resampled set, repeating $10{,}000$ times; reported CIs are percentile CIs on the per-case metric mean. PP@$k$ is computed by sorting nodes by predicted magnitude, taking the top $k\%$, sorting nodes by ground-truth magnitude, taking its top $k\%$, and reporting the Jaccard overlap; $\ewRMSE$ is the RMSE on the top $20\%$ of nodes by ground-truth magnitude.

\subsection{Direction-only and peak-normalised metrics on Womersley}

The headline VMR metrics defined above (PP@10 and $\ewRMSE$), together with the supplementary WSS $R^2$ negative control, are per-node relative or variance-normalised quantities and become ill-posed on the synthetic cylinder benchmark (\Cref{sec:womersley}). For that regime we report three replacement families. The \emph{direction-only success rate} at threshold $\theta$ is
\[
\PPdir{\theta} = \Pr\!\big[\angle(\mathbf{u}_\mathrm{pred},\mathbf{u}) \le \theta \,\big|\, \mathrm{HE}\big],
\]
i.e.\ the fraction of HE-mask nodes whose predicted velocity falls within an angular cone of the truth; the HE mask is the top-$20\%$-speed nodes, matching the training-time HE definition. The \emph{peak-normalised success rate} at tolerance $\delta$ is
\[
\mathrm{PP_{peak}}@\delta = \Pr\!\big[\|\mathbf{u}_\mathrm{pred} - \mathbf{u}\|/\max_i \|\mathbf{u}_i\| \le \delta \,\big|\, \mathrm{HE}\big],
\]
which differs from PP@$\delta$ by using the case-peak speed as the normaliser rather than the per-node speed, and is therefore stable to a multiplicative offset in predicted magnitude. The \emph{signed cosine} $\cosgn = \mathrm{median}_i[\hat{\mathbf{u}}_{\mathrm{pred},i}\!\cdot\!\hat{\mathbf{u}}_i]$ is the median per-node cosine similarity over the HE mask, and is sign-preserving (in contrast with the folded angle $\angle \in [0, 180^\circ]$); the \emph{per-node flip rate} $\fracflip = \Pr[\hat{\mathbf{u}}_\mathrm{pred}\!\cdot\!\hat{\mathbf{u}} < 0 \,|\, \mathrm{HE}]$ counts the fraction of HE nodes whose predicted direction is reversed relative to the truth.

\subsection{Per-node physics diagnostics}
\label{sec:physics_diag}

To probe physical content beyond the velocity-vector metrics we compute three per-node quantities on each prediction dump. The \emph{Spearman radial correlation} $\rho_\mathrm{rad} = \mathrm{Spearman}(\|\mathbf{u}\|, d_\mathrm{wall})$ correlates predicted speed with distance to the nearest wall node, a proxy for the local-radius coordinate of an idealised tube; a Poiseuille profile yields $\rho_\mathrm{rad} = +1$ in the limit of a long cylinder. The \emph{magnitude ratio} $r_{\mathrm{mag},i} = \|\mathbf{u}_{\mathrm{pred},i}\|/\|\mathbf{u}_i\|$ reports the per-node speed bias. The \emph{local divergence} of either field is estimated by a Gaussian-weighted least-squares Jacobian fit on each node's $16$-nearest-neighbour graph,
\[
\mathbf{J}_i = \arg\min_{\mathbf{J}}\sum_{j\in\mathcal{N}_k(i)} w_{ij}\,\|\mathbf{u}_j - \mathbf{u}_i - \mathbf{J}\,(\mathbf{x}_j-\mathbf{x}_i)\|^2,
\qquad w_{ij} = \exp\!\big(-\|\mathbf{x}_j - \mathbf{x}_i\|^2/h^2\big),
\]
with bandwidth $h$ set to the median nearest-neighbour distance; the estimated divergence is $(\nabla\!\cdot\!\mathbf{u})_i = \mathrm{tr}(\mathbf{J}_i)$. We report the mean of $|\nabla\!\cdot\!\mathbf{u}|$ normalised by the characteristic strain rate $U/L$, where $U$ is the mean HE speed and $L$ is the median nearest-neighbour spacing, on both the predicted field (a property of the network) and the ground-truth field (a property of CFD discretisation noise), so that the two can be directly compared (\Cref{sec:continuity}).

\subsection{Reproducibility}

The training, evaluation, and analysis pipelines are committed in the repository accompanying this paper.

\section{Data availability}
\label{sec:data-availability}

Pre-processed NPZ datasets for the four input schemas, the frozen train/validation/test splits, and the analytic Womersley generator are released with the companion repository under CC-BY-4.0. Patient-specific CFD data are referenced through per-case manifests and external-source download instructions: the raw VTU geometries and volumetric velocity fields are obtained from the Vascular Model Repository (\url{https://www.vascularmodel.com}) under its original licence, with redistribution restricted to that source. To support reproducibility without restricted patient-specific files, the repository includes analytic Womersley data generation and a quick-demo pipeline reproducing the core identifiability audit.

\section{Code availability}

Code and reproducibility materials---the input-schema audit implementation, four variant configurations, frozen splits, Womersley synthetic-data generator, advection--diffusion transfer demo (\Cref{sec:advdiff}), metric computation, bootstrap analysis, and the table and figure pipeline---are available at \url{https://github.com/Dyniel/flowgat-paper}.

\section{Acknowledgements}

We thank the Vascular Model Repository team for hosting the source models used for the patient-specific cohort.

\section{Funding}

This research did not receive any specific grant from funding agencies in the public, commercial, or not-for-profit sectors.

\section{Author contributions statement}

Daniel Cie\'slak: Conceptualization, Methodology, Software, Formal analysis, Investigation, Validation, Visualization, Data curation, Writing--original draft, Writing--review and editing. Andrzej Czy\.zewski: Supervision, Methodology, Resources, Project administration, Writing--review and editing. Both authors approved the final version.

\section{Declaration of generative AI and AI-assisted technologies in the manuscript preparation process}

During the preparation of this work, the authors used ChatGPT to support language editing, restructuring of journal-specific framing, and preparation of submission materials. After using this tool, the authors reviewed and edited the content as needed and take full responsibility for the content of the article.

\section{Additional information}
\noindent\textbf{Competing interests}

The authors declare no competing interests.

\clearpage
\raggedbottom

\renewcommand{\thefigure}{S\arabic{figure}}
\renewcommand{\thetable}{S\arabic{table}}
\renewcommand{\thesection}{S\arabic{section}}
\setcounter{figure}{0}
\setcounter{table}{0}
\setcounter{section}{0}

\section{Supplementary Information}
\renewcommand{\thefigure}{S\arabic{figure}}
\renewcommand{\thetable}{S\arabic{table}}
\setcounter{figure}{0}
\setcounter{table}{0}

\subsection{Proofs of the identifiability certificate}
\label{sec:proofs}

\begin{proof}[Proof of \Cref{thm:dir-id}]
Equation~\eqref{eq:closedform} gives the leading-order decomposition of the velocity into
an axial component parallel to \(\unit{T}(s)\) and a residual bounded by the rod small
parameters. Dividing by \(\|\mathbf{u}\|\) and applying the elementary angle perturbation
bound for a vector plus a transverse residual gives \Cref{eq:dir-bound}. In the limit
\(\epsilon=De=0\), the residual vanishes and the statement is exact for straight
Poiseuille or Womersley flow.

For magnitude, \Cref{eq:mag-factor} is the norm of \Cref{eq:closedform}. The mesh
determines \(R(s)\) and the tangent frame, but it does not determine the proximal waveform
\(Q(t)\), the cycle phase, or outlet impedance. Fixing \(\Omega\) and varying \(Q(t)\)
therefore produces two admissible Navier--Stokes initial-boundary value problems with
identical \(\mathcal{A}_{\mathrm{geo}}\) and different speed fields. Any prediction rule whose
test-time output is \(\mathcal{A}_{\mathrm{geo}}\)-measurable must assign the same output to
both problems; it cannot recover both magnitude fields with asymptotically vanishing error.
This is an information obstruction, not an optimisation failure.

For the signed-vector statement, use the linearity of the leading-order profile in \(Q(t)\).
The involution \(Q(t)\mapsto -Q(t)\) maps \(\mathbf{u}_\star\) to \(-\mathbf{u}_\star\)
while preserving \(\Omega\), \(R(s)\), unsigned geometric axes, and
\(\|\mathbf{u}_\star\|\). The same symmetry is realised in the analytic Womersley
benchmark by a \(\pi\) phase shift of the driving pressure gradient. Thus the declared input
has at least two physically admissible preimages that differ by orientation. Any deterministic
or randomized predictor receiving only that input can select at most one representative of
the quotient class, and no architecture can remove the ambiguity without receiving an
additional sign-fixing channel.
\end{proof}

\begin{proof}[Proof of \Cref{prop:quotient}]
(i) The map $\unit{T}:\Omega\to\mathbb{S}^2/\{\pm 1\}$ is a measurable
functional of the mesh (it is the unit eigenvector of the local
medial-axis covariance), so $\pi_{\mathrm{line}}\circ Y$ is
$\mathcal{A}_{\mathcal{I}_{\mathrm{geo}}}$-measurable by
\Cref{eq:dir-bound}. The non-measurability of $\|Y\|$ and
$\hat{Y}$ is \Cref{thm:dir-id}\,(b)--(c). (ii) The
$Q\mapsto -Q$ involution preserves $\|\mathbf{u}_\star\|$ and
$\mathcal{I}_{\mathrm{geo}}$, so $\pi_{\mathrm{sgn}}\circ Y$ is the
finest $\mathcal{A}_{\mathcal{I}_{\mathrm{geo+mag}}}$-measurable
coarsening of $Y$. (iii) The two preimages $f_\pm$ realise the
quotient explicitly; on Womersley these are the two cycle phases at
which $Q(t)$ changes sign.
\end{proof}

\begin{proof}[Proof of \Cref{prop:arch-invariance}]
Bias--variance decomposition relative to the conditional expectation
gives, for every $\mathcal{A}_{\mathcal{I}}$-measurable $\Phi$,
$\mathbb{E}\|\Phi(\mathcal{I})-Y\|^2 =
\mathbb{E}\|\mathbb{E}[Y\mid\mathcal{I}]-Y\|^2 +
\mathbb{E}\|\Phi(\mathcal{I})-\mathbb{E}[Y\mid\mathcal{I}]\|^2
\ge \mathcal{R}^{\star}_{\mathcal{I}}$. Every element of
$\mathcal{H}_{\mathcal{I}}$ is, by hypothesis,
$\mathcal{A}_{\mathcal{I}}$-measurable; randomised training adds an
independent randomisation $\omega$ but
$\mathbb{E}_\omega[\Phi_{\hat\theta_N(\omega)}(\mathcal{I})]$ is
still $\mathcal{A}_{\mathcal{I}}$-measurable, so the same bound
applies in expectation. Positivity of
$\mathcal{R}^{\star}_{\mathcal{I}_{\mathrm{geo}}}$ and
$\mathcal{R}^{\star}_{\mathcal{I}_{\mathrm{geo+mag}}}$ follows from
\Cref{thm:dir-id}\,(b)--(c) and \Cref{prop:quotient}\,(iii): an
$\mathcal{I}$-collision $(f_+,f_-)$ with $Y(f_+)\ne Y(f_-)$ forces
$\mathrm{Var}[Y\mid\mathcal{I}]>0$ on a positive-measure subset of
$\mathcal{F}$.
\end{proof}

\begin{proof}[Proof of \Cref{lem:sign}]
(i) By \Cref{prop:quotient}\,(ii), the two orientations $\sigma=\pm 1$
project to the same $X$ but to opposite high-energy signs of
$\mathbf{u}\cdot\unit{T}$. Conditional on $X$, predicting the sign is
therefore a binary decision problem with Bayes error
$\min\{\pi_+(X),\pi_-(X)\}$; the rod residual perturbs the
two posteriors by at most $\mathcal{O}(\epsilon+De)$.
(ii) For the orientation-mixed target on the HE mask, the
$L^2$ excess risk of any
$\mathcal{A}_{\mathcal{I}_{\mathrm{geo+mag}}}$-measurable predictor
against $Y$ satisfies
$\mathbb{E}\|\Phi^{\star}-Y\|^2_{\mathrm{HE}}\ge
\mathrm{Var}[\sigma\mid X]\cdot\|Y\|_{\mathrm{HE}}^2$, and
$\mathrm{Var}[\sigma\mid X]=4\pi_+\pi_-$ for a Bernoulli$(\pi_+)$
variate supported on $\{-1,+1\}$.
(iii) Apply Fano's inequality to the binary orientation decision
on each connected component; independence across components yields
the sum form.
(iv) Every element of $\mathcal{H}_{\mathcal{I}_{\mathrm{geo+mag}}}$
is $\mathcal{A}_{\mathcal{I}_{\mathrm{geo+mag}}}$-measurable by
construction, so \Cref{prop:arch-invariance} applies and the bound
transfers.
\end{proof}

\subsection{Confound audit for the magnitude oracle-probe result}
\label{sec:supp_sanity}
\begin{table}[H]
\small
\centering
\caption{\textbf{Sanity-check audit of the magnitude-without-direction
result.} The phenomenon that a magnitude oracle probe (\magonly{}) lands
strictly below the geometry-only-input baseline (\geoschema{}) on the
velocity-vector endpoints is the most counter-intuitive empirical
result of the audit---an estimator-level phenomenon that sits below
the Bayes-risk ordering $\mathcal{R}^{\star}_{\mathcal{I}_{\mathrm{geo+mag}}}\le\mathcal{R}^{\star}_{\mathcal{I}_{\mathrm{geo}}}$
rather than contradicting it. Each row below is a confound we ruled out by
design or by replication, with the matching control evidence and the
section in which it is reported. The reading is structural: under the
input schema $\mathcal{I}_{\mathrm{geo+mag}}$ the
optimiser is asked to deposit an exact amplitude into a vector field
for which $\mathcal{A}_{\mathcal{I}}$ provides no asymmetric orienting
frame (\Cref{prop:quotient}\,(ii), \Cref{lem:sign}), so the
fitted predictor is pulled towards inconsistent local minima; the
same amplitude becomes fully consumable once a per-node direction
channel is supplied (\dironly{} vs.\ \dirmagprobe{}).}
\label{tab:sanity}
\begin{tabularx}{\linewidth}{p{0.30\linewidth}X}
\toprule
Potential confound & Control evidence \\
\midrule
Different optimiser, seeds, or train/val/test split between variants.
& Identical AdamW schedule, identical augmentation, identical three
seeds $\{1337, 2026, 777\}$, identical $37/4/5$ split (Methods,
\Cref{tab:efficiency}). The four variants differ only in the two
input channels listed in \Cref{tab:variants}. \\
\midrule
Magnitude channel normalisation drift across variants. &
$\mathbf{x}[8]$ is normalised by the same characteristic speed
$u_\mathrm{char}=1$~m/s in every variant; the geometric proxy uses
the same characteristic length $L_\mathrm{char}=25$~mm; numerical
ranges of $\mathbf{x}[8]$ overlap on the unit interval across the
four variants (Methods). \\
\midrule
Hidden target leakage through channels other than the declared
oracle. & All non-oracle channels are mesh-derived
($\mathcal{G}(\Omega)$) and identical across variants; the no-slip
mask zeros the loss on wall nodes; the oracle is the only
$\mathcal{A}_{\mathcal{I}_{\mathrm{geo+mag}}}$-enlargement
relative to $\mathcal{A}_{\mathcal{I}_{\mathrm{geo}}}$. \\
\midrule
Direction proxy is too crude (case-global PCA), so the effect is a
proxy-quality artefact. & Replacing the global PCA tangent with a
per-node Frenet tangent from a medial-axis skeleton
(\centerlinegeo{}) does not change the answer: $\cosgn$
$0.799$ vs $0.798$, $\fracflip$ $0.136$ vs $0.159$, angle median
$35.4^\circ$ vs $34.8^\circ$ (P1; \Cref{sec:robustness}, Supp.\
\Cref{sec:supp_centerline}). \\
\midrule
Effect is specific to the GATv2 attention backbone. &
A structurally distinct GraphSAGE backbone (no attention, no
edge bias, no hard-no-slip head) reproduces both the ceiling and
the floor of the asymmetric pattern within seed s.d.\ (P2;
Supp.\ \Cref{sec:supp_sage}). \\
\midrule
Effect is an artefact of in-vivo cohort complexity or of the
patient-specific geometry. & The same asymmetric pattern persists
on the analytic Womersley benchmark, where the geometric proxy
$\hat{z}$ is exact by construction and the only residual ambiguity
is the $\mathbb{Z}_2$ orientation predicted by
\Cref{prop:quotient} (P3; \Cref{sec:womersley}). \\
\midrule
Sign degeneracy is the trivial ``model learned a fixed polarity and
gets caught at flow reversal'' artefact. & Per-phase flip rate on
Womersley is $0.39$ in the forward half-cycle and $0.33$ in the
reverse half-cycle---statistically indistinguishable---which
falsifies the phase-locked reading and leaves the
input-schema reading of \Cref{lem:sign} as the supported mechanism
(P4; \Cref{fig:phase_cos}). \\
\bottomrule
\end{tabularx}
\end{table}

\subsection{Per-case and per-pathology breakdown}
\label{sec:supp_percase}

\begin{figure}[H]
\centering
\includegraphics[width=0.76\linewidth]{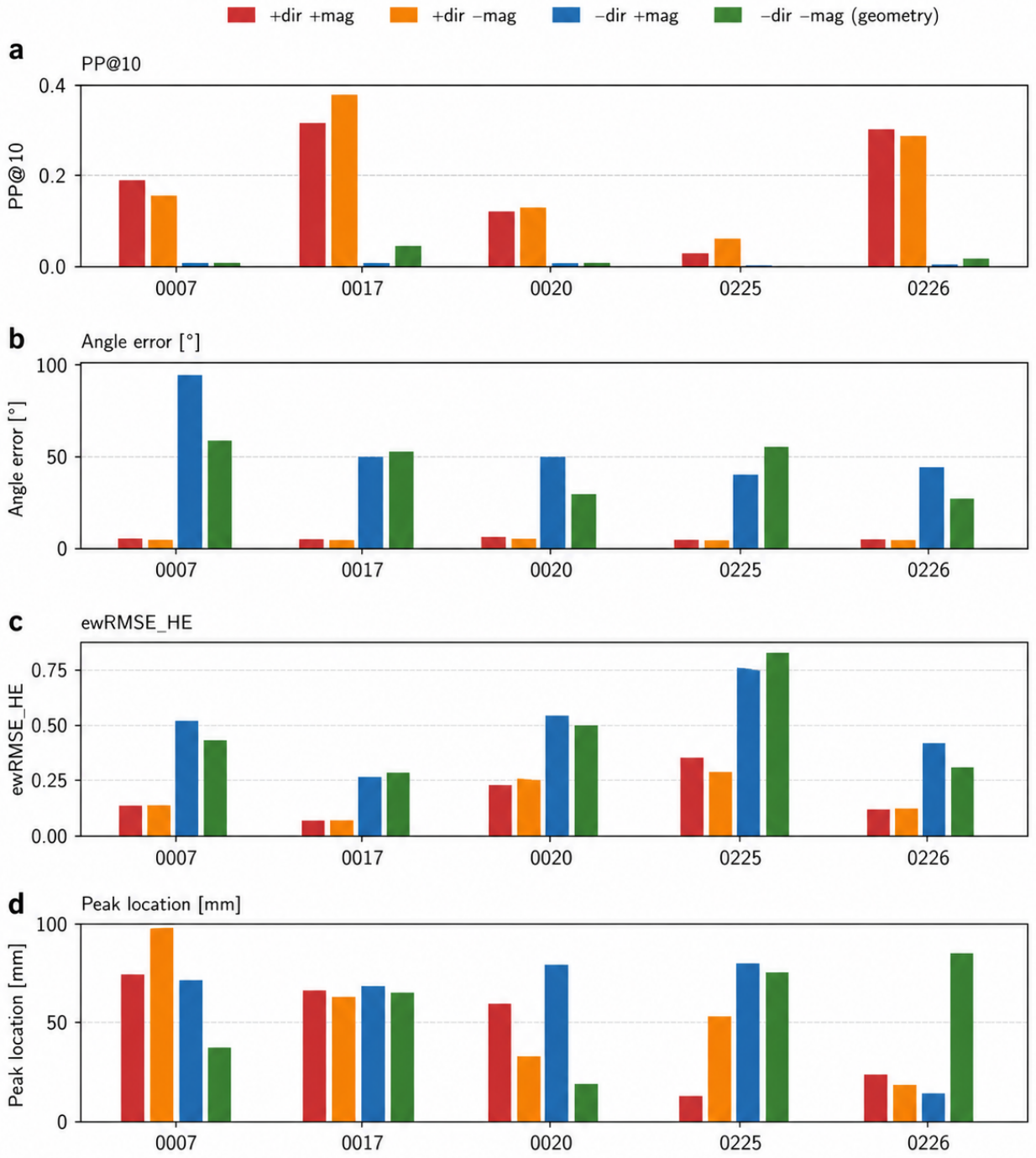}
\caption{\textbf{Per-case test-set metrics across the four variants.}
Five test cases (columns) $\times$ four metrics (rows); each panel
plots the 3-seed mean and s.d. The \dironly{}/\dirmagprobe{} cluster
versus \magonly{}/\geoschema{} cluster is preserved on every case for the
velocity-vector metrics.}
\label{fig:fig3}
\end{figure}

\begin{figure}[H]
\centering
\includegraphics[width=0.90\textheight]{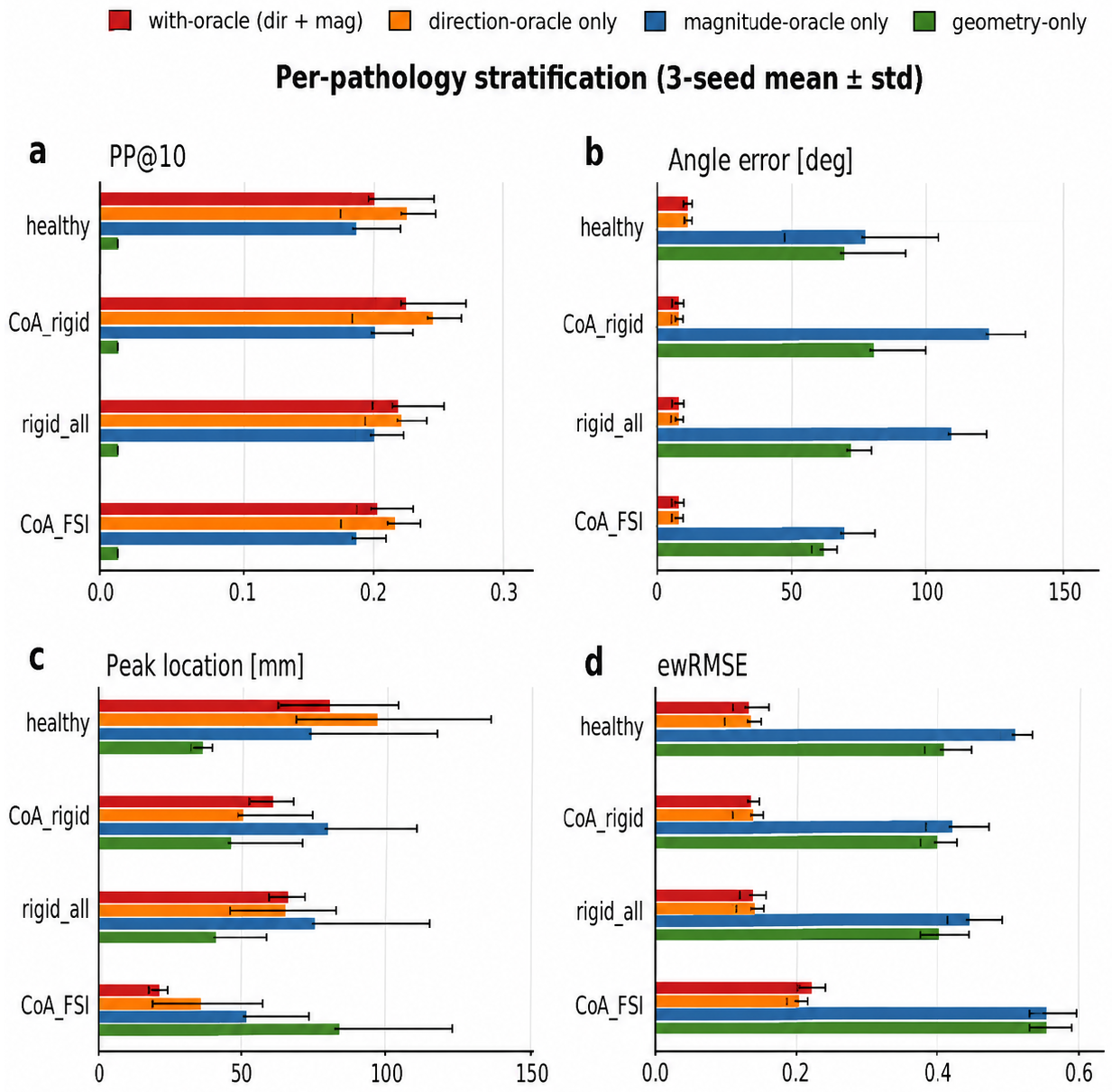}
\caption{\textbf{Per-pathology stratification.} Healthy
(\texttt{0007}, $n=1$), rigid coarctation (\texttt{0017},
\texttt{0020}, $n=2$), and FSI coarctation (\texttt{0225},
\texttt{0226}, $n=2$). The direction--magnitude asymmetry is
preserved across all three strata.}
\label{fig:fig5}
\end{figure}

\subsection{Direction--magnitude identifiability decomposition}
\label{sec:supp_decomp}
\begin{figure}[H]
\centering
\includegraphics[width=0.65\linewidth]{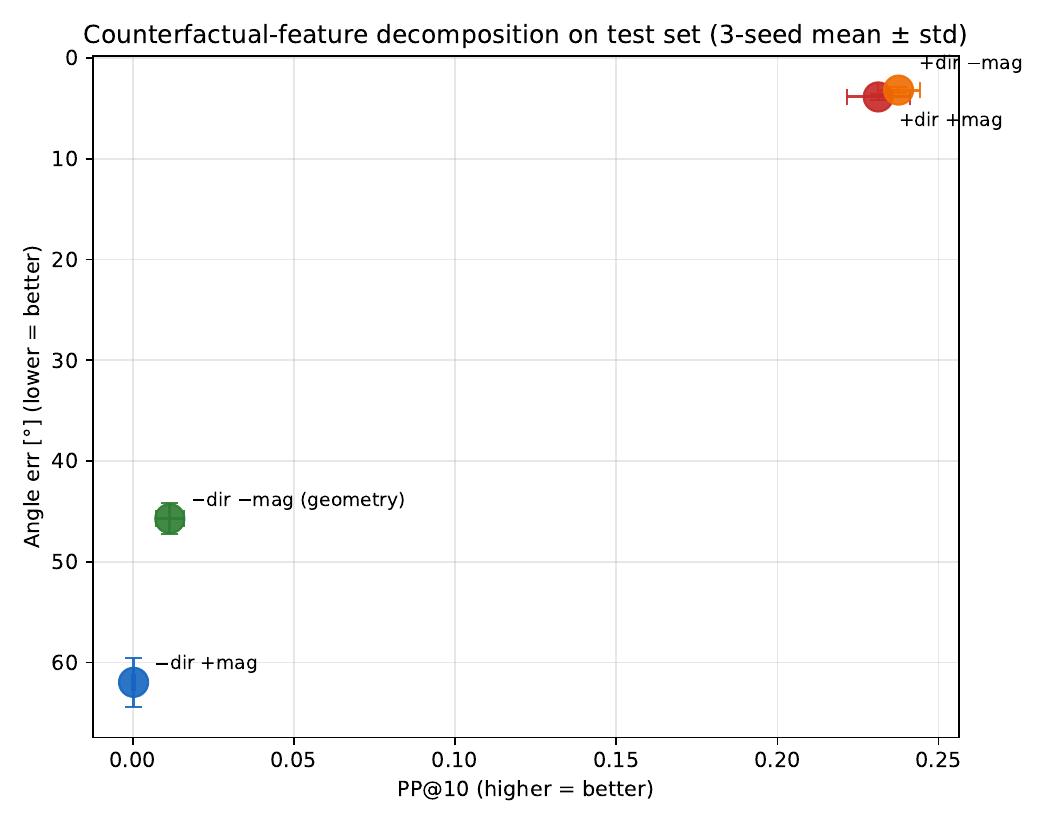}
\caption{\textbf{Direction--magnitude identifiability decomposition.}
Additive-with-interaction decomposition of PP@10 across the four
input schemata. The dominant direction oracle probe main effect, the
positive magnitude-only main effect (non-identifying under
$\mathcal{I}_{\mathrm{geo+mag}}$), and the large negative interaction
term that cancels the magnitude-only effect once a per-node direction
frame is supplied are all visible.}
\label{fig:fig6}
\end{figure}

\subsection{Architecture $\times$ domain replication and Dean-number sensitivity}
\label{sec:supp_archdomain}
\begin{table}[H]
\centering
\small
\setlength{\tabcolsep}{4pt}
\renewcommand{\arraystretch}{1.08}
\resizebox{\linewidth}{!}{%
\begin{tabular}{llcccc}
\toprule
run & variant & $\PPdir{10}$ & $\cosgn$ & $\fracflip$ & angle\textsuperscript{med} ($^\circ$) \\
\midrule
SAGE $\times$ Cosserat  & \dirmagprobe{}      & $0.998 \pm 0.013$ & $+0.9995 \pm 0.000$ & $0.000$ & $1.60 \pm 0.67$ \\
                        & \dironly{}       & $0.999 \pm 0.009$ & $+0.9996 \pm 0.000$ & $0.000$ & $1.38 \pm 0.63$ \\
                        & \magonly{}       & $0.033 \pm 0.066$ & $+0.209 \pm 0.614$ & $\mathbf{0.383 \pm 0.363}$ & $75.2 \pm 41.9$ \\
                        & \geoschema{}        & $0.029 \pm 0.061$ & $+0.249 \pm 0.597$ & $\mathbf{0.363 \pm 0.353}$ & $73.0 \pm 40.6$ \\
\midrule
SAGE $\times$ sUbend    & \dirmagprobe{}      & $0.874 \pm 0.115$ & $+0.997 \pm 0.002$ & $0.000$ & $3.93 \pm 1.69$ \\
                        & \dironly{}       & $0.908 \pm 0.092$ & $+0.998 \pm 0.002$ & $0.000$ & $3.51 \pm 1.54$ \\
                        & \magonly{}       & $0.171 \pm 0.097$ & $+0.708 \pm 0.292$ & $\mathbf{0.235 \pm 0.149}$ & $41.0 \pm 22.0$ \\
                        & \geoschema{}        & $0.150 \pm 0.085$ & $+0.679 \pm 0.288$ & $\mathbf{0.225 \pm 0.151}$ & $44.0 \pm 20.9$ \\
\midrule
FlowGAT $\times$ Cosserat (no BC) & \dirmagprobe{}  & $0.992 \pm 0.049$ & $+0.999 \pm 0.001$ & $0.000$ & $1.75 \pm 0.96$ \\
                        & \dironly{}       & $0.987 \pm 0.080$ & $+0.999 \pm 0.002$ & $0.000$ & $1.70 \pm 1.39$ \\
                        & \magonly{}       & $0.037 \pm 0.092$ & $+0.184 \pm 0.644$ & $\mathbf{0.387 \pm 0.369}$ & $77.1 \pm 44.6$ \\
                        & \geoschema{}        & $0.025 \pm 0.072$ & $+0.201 \pm 0.628$ & $\mathbf{0.373 \pm 0.359}$ & $76.3 \pm 43.2$ \\
\bottomrule
\end{tabular}%
}
\caption{\textbf{Phase E8: four-variant direction-identifiability results across two architectures, two domains and two BC mechanisms.} All values are mean $\pm$ SD across $3$ seeds (test split). The step-function — direction-bearing variants at $\PPdir{10} \approx 1.0$ (Cosserat) or $\approx 0.9$ (sUbend); no-direction variants collapsed to chance ($0.03$--$0.17$) — holds across all three conditions. sUbend collapse floors are higher than Cosserat ($\approx 0.15$--$0.17$ vs.\ $\approx 0.03$) owing to residual direction information in the realistic-tube geometry (\Cref{sec:e8_robustness}); this is consistent with the Cosserat-residual bound of \Cref{thm:dir-id} and is also present in the FlowGAT $\times$ sUbend audit reported in \Cref{sec:womersley} context.}
\label{tab:e8_direction}
\end{table}

\begin{figure}[H]
\centering
\includegraphics[width=0.98\linewidth]{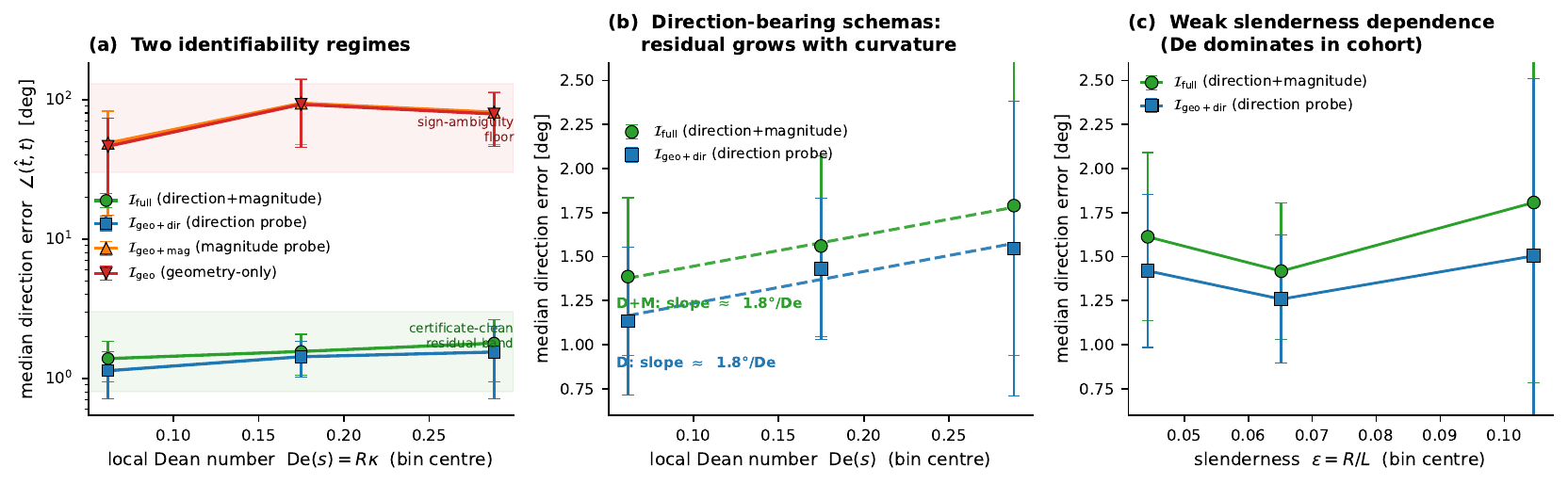}
\caption{\textbf{Sensitivity of the direction error to the Dean number $De$ and
slenderness $\varepsilon$ (SAGE$\times$Cosserat sweep, test split, mean $\pm$ SD
over $3$ seeds).} \textbf{(a)} Median per-node angular error against $De$ tercile
for all four input schemas (log scale). Two regimes separate by two orders of
magnitude: direction-bearing schemas occupy a certificate-clean residual band of
$\sim\!1$--$2^\circ$, the no-direction schemas a sign-ambiguity floor of
$\sim\!50$--$100^\circ$. \textbf{(b)} The two direction-bearing schemas on a
linear scale; the error grows monotonically with $De$ with a fitted slope of
$\approx 1.8^\circ$ per unit $De$ and near-zero intercept, matching the
$C_{\mathrm{dir}}(\epsilon+De)$ residual of \Cref{thm:dir-id}(a). \textbf{(c)} The
same schemas against $\varepsilon$: the slenderness dependence is weak and
non-monotone, so $De$ (curvature) dominates the residual in this cohort. Data:
\texttt{phase\_E8\_package/sage\_cosserat/stratified\_by\_\{de,eps\}.csv};
regeneration script \texttt{figures/gen\_fig\_sensitivity.py}.}
\label{fig:sensitivity}
\end{figure}

\subsection{Advection--diffusion transfer check: closed-form factorisation}
\label{sec:supp_advdiff}

The transfer check of main-text \Cref{sec:advdiff} applies the seven-step audit to the
steady one-dimensional advection--diffusion equation $a\,u_x = D\,u_{xx}$ on $[0,1]$,
with Dirichlet data $u(0)=u_L$, $u(1)=u_R$ and an unknown inflow boundary condition.
The exact solution is
\[
u(x) = u_L + (u_R-u_L)\,\psi(x;\mathrm{Pe}),
\qquad
\psi(x;\mathrm{Pe}) = \frac{e^{\mathrm{Pe}\,x}-1}{e^{\mathrm{Pe}}-1},
\qquad
\mathrm{Pe} = \frac{a}{D},
\]
where $\mathrm{Pe}$ is the P\'eclet number. Centring the profile, $u-\tfrac12(u_L+u_R)$,
factorises the target into three pieces that map one-to-one onto the velocity case:
a \emph{shape} $\phi(\cdot;|\mathrm{Pe}|)$ fixed by the coefficient magnitude
$|\mathrm{Pe}|$ (the role played by geometry), an \emph{amplitude}
$|A| = |u_R-u_L|$ fixed by the boundary condition (the role played by the inflow
waveform), and an \emph{orientation} $s=\mathrm{sgn}(u_R-u_L)=\pm1$---which end of the
domain carries the boundary layer---that the magnitude of $\mathrm{Pe}$ leaves
ambiguous up to a $\mathbb{Z}_2$ sign. Steps one--four of the certificate therefore
predict, before any fitting, the same three identifiability classes
(input-measurable shape, boundary-measurable amplitude, quotient-identifiable
orientation) as the flow problem.

For steps five--seven we draw an ensemble of $6000$ random admissible boundary
conditions (sampling $\mathrm{Pe}$, $u_L$ and $u_R$ over the admissible ranges) and
compute the Bayes-optimal predictor under each of the four input schemas. The result
(\Cref{fig:advdiff}, main text) reproduces the flow audit term for term: the
geometry-only schema recovers $\phi$ but must commit to a fixed sign and sits on the
$\mathbb{Z}_2$ floor ($\mathrm{PP_{dir}}@10^\circ=0.50$, flip rate $0.50$); a magnitude
channel collapses the amplitude error but leaves the flip rate at $0.50$, as a
non-orienting channel must; an orientation channel instead removes the flip
($\mathrm{PP_{dir}}@10^\circ=1.00$, $\fracflip=0$) up to a $1$--$2^\circ$ shape residual;
only the full schema resolves both amplitude and orientation. The target shares none
of the Cosserat machinery, so the agreement confirms that the audit procedure, not the
particular flow factorisation, is what transfers.

\subsection{Disclosure checklist for geometry-conditioned flow surrogates}
\label{sec:supp_checklist}

The identifiability certificate is a statement about the declared input schema, not
about any particular network or loss. It therefore translates into a short,
schema-level disclosure that lets a reader decide \emph{a priori} whether a reported
velocity-field accuracy is attainable from the stated inputs or is instead evidence
of an undeclared oracle channel. We recommend that subsequent surrogates report the
following seven items.

\begin{enumerate}
\item \textbf{Per-channel input inventory.} List every per-node and global input
channel with its physical meaning, and mark each as geometric (frame-covariant) or
dynamical. The certificate applies to the geometric subset; magnitude and orientation
must be sought among the remaining channels.
\item \textbf{Orientation provenance.} State whether a signed direction enters the
inputs (e.g.\ a proximal inflow waveform, a cycle-phase scalar, a 4D-flow streamline
tangent) or whether all orientation must be inferred from unsigned geometry. Absence
of a signed channel triggers the $\mathbb{Z}_2$ sign obstruction of
\Cref{lem:sign}.
\item \textbf{Magnitude provenance.} State whether absolute flow magnitude is supplied
(an inflow rate $Q(t)$, a Reynolds or Womersley number) or only the dimensionless
geometry; the magnitude obstruction of \Cref{thm:mag-obs} applies in the latter case.
\item \textbf{Direction/magnitude variant sweep.} Report the four-variant decomposition
(\geoschema, plus direction-, magnitude-, and direction+magnitude oracle probes) so
that the geometric-only floor and the oracle ceiling are both visible. A single
end-to-end accuracy number cannot separate the two.
\item \textbf{Signed-cosine and flip-rate metrics.} Report $\cosgn$ and the per-node
flip rate $\fracflip$ in addition to angular error, since an unsigned metric hides the
sign degeneracy that the certificate predicts.
\item \textbf{Boundary-condition masking.} State how wall and inlet/outlet nodes are
masked in the loss, so a reader can rule out boundary leakage as the source of an
apparent interior-field guarantee.
\item \textbf{Slenderness and curvature range.} Report the cohort ranges of $\varepsilon$
and $De$, because the certificate's residual bound $C_{\mathrm{dir}}(\epsilon+De)$ grows
with curvature; an accuracy claim is only meaningful relative to the regime in which
it was measured.
\end{enumerate}

Items 1--3 are schema declarations that can be checked before any training; items 4--7
are the minimal measurements that expose the predicted floors. Together they let a
reviewer distinguish ``the model learned the flow'' from ``the model read the tangent
and copied the magnitude'' without rerunning the experiment.

\subsection{Interventions that can buy a guarantee, and their theoretical ceilings}
\label{sec:supp_piml_routes}

Because the obstructions of \Cref{lem:sign,thm:mag-obs} concern the declared input set
rather than the loss, only interventions that enlarge the input with an asymmetric
(orientation- or magnitude-bearing) channel can remove them; loss-side regularisers
leave the input ambiguity intact. We list the concrete routes together with the upper
bound the certificate places on their attainable gain.

\begin{enumerate}
\item \textbf{Proximal inflow waveform $Q(t)$.} Supplies absolute magnitude and a
temporal orientation. This lifts the magnitude obstruction of \Cref{thm:mag-obs}
exactly, since $Q(t)$ fixes the scale that geometry alone cannot. The residual angular
error is still bounded below by the curvature term $C_{\mathrm{dir}}(\epsilon+De)$ of
\Cref{thm:dir-id}: magnitude information does not improve orientation.
\item \textbf{Cycle-phase scalar.} A single signed phase variable breaks the
$\mathbb{Z}_2$ degeneracy of \Cref{lem:sign} by distinguishing systolic forward flow
from the diastolic branch. Its ceiling is the same direction residual as a full signed
tangent when the cohort is dominated by a single dominant flow axis.
\item \textbf{Signed centreline tangent from 4D-flow MRI streamlines~\cite{ferdian2020}.}
Provides a per-node signed direction, collapsing the angular error to the oracle band
of $\sim\!1$--$2^\circ$ predicted by \Cref{thm:dir-id}(a); it does not by itself fix
magnitude, so it must be paired with route~1 for a full velocity guarantee.
\item \textbf{Inflow-aware boundary loss~\cite{westerhof2009arterial}.} Encodes the
signed inlet flux as a soft boundary constraint. It behaves as a weak version of
route~1, with a gain that degrades as the boundary signal is diluted over the interior;
the certificate caps its benefit at the magnitude obstruction it partially addresses.
\item \textbf{Helmholtz-projected divergence-free residual~\cite{liu2024multiresolution}
(loss-side, for contrast).} A representative loss-only intervention. Because it adds no
asymmetric input channel, \Cref{lem:sign,thm:mag-obs} guarantee it cannot remove either
obstruction; its expected gain on the certificate-relevant axes is zero, which the
no-flow transfer check corroborates.
\end{enumerate}

Routes 1--4 each enlarge the input schema and therefore have nonzero ceilings; route~5
illustrates why a loss-side fix, however physically motivated, cannot exceed the bound
set by the declared inputs.

\subsection{Training-efficiency telemetry}
\label{sec:efficiency}

\begin{table}[H]
\small
\centering
\caption{Per-(variant, seed) training-efficiency telemetry. Total
parameter count is $842{,}755$ for every run; peak GPU memory is
$\sim\!25.4$ GB. ``best epoch'' is the epoch with maximum
validation PP@10.}
\label{tab:efficiency}
\begin{tabular}{llrrrr}
\toprule
Variant & seed & best epoch & total epochs & best val PP@10 & train time (h) \\
\midrule
\dirmagprobe{}  & 1337 & 120 & 720  & 0.2404 & 4.96 \\
\dirmagprobe{}  & 2026 & 170 & 770  & 0.2174 & 5.41 \\
\dirmagprobe{}  & 777  & 180 & 780  & 0.2531 & 5.35 \\
\dironly{}   & 1337 & 340 & 940  & 0.2199 & 6.48 \\
\dironly{}   & 2026 & 140 & 740  & 0.2391 & 5.16 \\
\dironly{}   & 777  & 280 & 880  & 0.2676 & 6.09 \\
\magonly{}   & 1337 & 420 & 1020 & 0.0058 & 6.92 \\
\magonly{}   & 2026 &  40 & 640  & 0.0092 & 4.37 \\
\magonly{}   & 777  & 190 & 790  & 0.0041 & 5.41 \\
\geoschema{}    & 1337 & 980 & 1500 & 0.0243 & 10.23 \\
\geoschema{}    & 2026 & 650 & 1250 & 0.0351 & 8.61  \\
\geoschema{}    & 777  & 730 & 1330 & 0.0431 & 9.16  \\
\bottomrule
\end{tabular}
\end{table}

The \magonly{} runs converge \emph{quickly} to a poor plateau, not slowly to a good one; the \geoschema{} runs converge slowly to a moderate plateau. The pattern is consistent with the input-schema reading: magnitude information without a direction frame is mis-fittable, while geometry-only inputs are slow but informative.

\subsection{Peak-localisation statistic dependence}
\label{sec:peakloc}

\begin{table}[H]
\small
\centering
\caption{Peak-localisation discrepancy (Euclidean distance, mm,
between predicted and true peak-magnitude nodes) is highly sensitive
to the choice of summary statistic on the test cohort. We therefore
\emph{do not} place this metric in the headline.}
\label{tab:peakloc}
\begin{tabular}{lll}
\toprule
Statistic & \dirmagprobe{} & \geoschema{} \\
\midrule
median-of-medians (per seed)  & $49.9 \pm 11.7$  & $37.4 \pm 0.7$ \\
mean-of-means    (per seed)   & $47.2$           & $58.2$         \\
\bottomrule
\end{tabular}
\end{table}

The reversal under the mean statistic is driven by two FSI cases (\texttt{0225}, \texttt{0226}) on which \geoschema{} predicts the peak in the distal descending aorta, $\sim\!78$--$90$~mm from the true (proximal-coarctation) peak. We disclose this explicitly and decline to claim peak-localisation as a positive finding for \geoschema{}. \Cref{fig:fig4} overlays these per-case peak predictions on the test-set 3-D meshes.

\subsection{Wall-shear-stress metrics (relative contrasts only)}
\label{sec:wss}

The wall-shear-stress coefficient of determination is negative for all four variants on the present cohort (\Cref{tab:wss_supp}). This is a known difficulty for cardiovascular GNN surrogates~\cite{tabe2026pignn}, not a finding of our audit---it reflects that none of our four interventions is a clinical WSS predictor on this $n = 5$ cohort. WSS MAE and $R^2$ are reported here as relative contrasts only.

\begin{table}[H]
\small
\centering
\begin{tabular}{lll}
\toprule
Variant & WSS MAE (Pa) & WSS $R^2$ \\
\midrule
\dirmagprobe{} & $14.32 \pm 0.39$ & $-9.7 \pm 0.5$ \\
\dironly{}  & $12.30 \pm 0.83$ & $-7.1 \pm 1.2$ \\
\magonly{}  & $\phantom{0}9.07 \pm 1.23$ & $-3.7 \pm 1.9$ \\
\geoschema{}   & $13.86 \pm 1.32$ & $-14.1 \pm 2.7$ \\
\bottomrule
\end{tabular}
\caption{WSS metrics by oracle-channel intervention
(test cohort, $3$ seeds $\times$ $5$ cases). $R^2 < 0$ uniformly:
the constant predictor outperforms every variant on WSS on this
small cohort.}
\label{tab:wss_supp}
\end{table}

\subsection{Validation of the per-node divergence estimator}
\label{sec:supp_divergence}

The continuity negative control (\Cref{sec:continuity}) compares the normalised
mean absolute divergence of the predicted and ground-truth velocity fields. Its
conclusion---that apparent mass conservation on the patient cohort is a learned
data statistic rather than an imposed physical law---rests on the divergence
\emph{estimator} being trustworthy: it must not manufacture spurious divergence
from a solenoidal field, and it must recover true divergence when it is present.
We verify both on synthetic point clouds where the divergence is known
analytically, using the estimator exactly as defined in \Cref{sec:physics_diag}
(Gaussian-weighted least-squares Jacobian on the $16$ nearest neighbours, with
bandwidth equal to the median nearest-neighbour distance, reporting
$\mathrm{tr}\,\mathbf{J}$ normalised by $U/L$). No network and no measured data
enter this check; it isolates the estimator.

\Cref{fig:divergence_validation} reports three experiments. \textbf{(a)
Calibration.} We take a divergence-free Arnold--Beltrami--Childress base field
and add an isotropic source term $c\,\mathbf{x}$ of exact divergence $3c$,
sweeping $c$ through positive and negative values. The estimated normalised
divergence falls on the line $y=x$ with fitted slope $1.000$ and intercept
$-2\times10^{-4}$: the estimator is unbiased and recovers imposed divergence
across the tested range. \textbf{(b) False-positive floor.} On the pure
divergence-free field with additive Gaussian velocity noise of magnitude
$\sigma=f\,U$, the spurious normalised $|\nabla\!\cdot\!\mathbf{u}|$ is small and
rises smoothly with noise---$0.01$ at $f=0$, $0.03$ at $f=2\%$, $0.12$ at
$f=10\%$---and the uniform-box and irregular-tube clouds agree to within their
spread, confirming the floor is a property of noise rather than of sampling
geometry. \textbf{(c) Neighbourhood robustness.} At a fixed $2\%$ noise level the
floor is flat in the neighbour count $k$ ($0.027$--$0.031$ for $k\in[8,32]$),
so the $k=16$ used in the paper is not a tuned choice. Two consequences follow
for \Cref{sec:continuity}. First, the VMR predicted-field residual
($\approx0.003$) sits \emph{below} even the $f=0$ floor, i.e.\ within estimator
resolution of divergence-free---so the network's VMR field genuinely looks
mass-conserving at this resolution. Second, the Womersley predicted-field
residual ($\approx10^{-1}$) is an order of magnitude above the $10\%$-noise
floor, so it is a real loss of continuity and not an estimator artefact. The
contrast on which the negative control turns is therefore robustly above the
estimator's own noise.

\begin{figure}[H]
\centering
\includegraphics[width=0.98\linewidth]{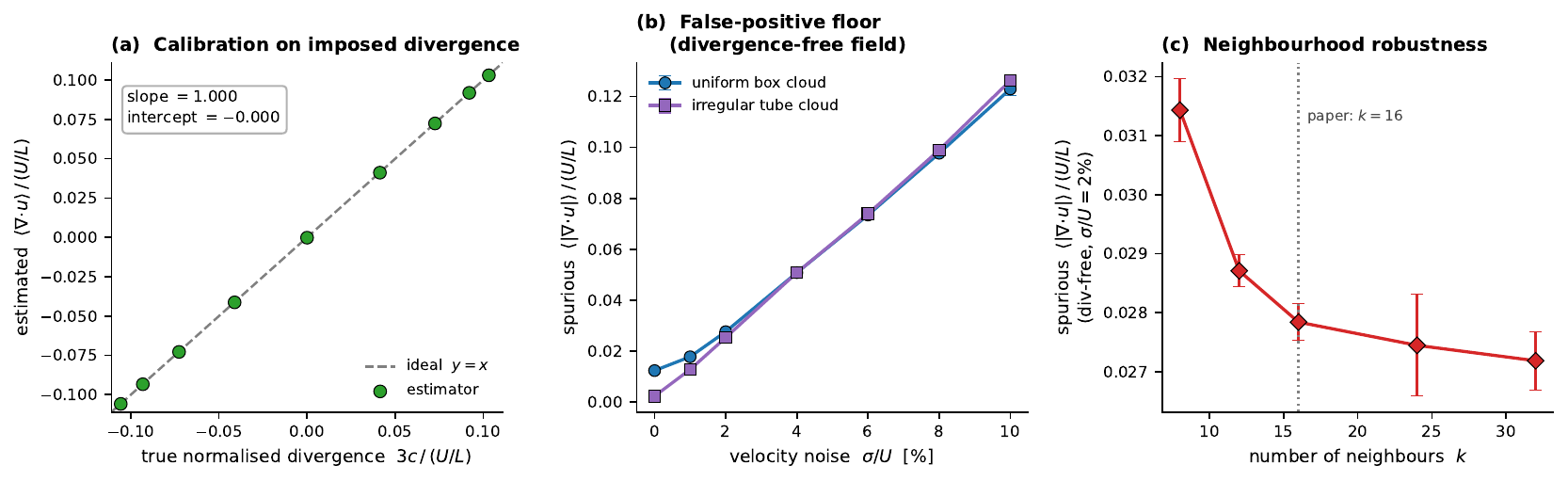}
\caption{\textbf{Validation of the per-node divergence estimator on synthetic
fields with known divergence.} \textbf{(a)} Calibration against an imposed
isotropic source of exact divergence $3c$ added to a divergence-free base field;
the estimator recovers it on $y=x$ (slope $1.000$, intercept $-2\times10^{-4}$).
\textbf{(b)} Spurious $|\nabla\!\cdot\!\mathbf{u}|$ on a divergence-free field
versus velocity-noise level, for a uniform-box and an irregular tube-like cloud;
the false-positive floor is small and noise-controlled. \textbf{(c)} The same
floor versus neighbour count $k$ at fixed $2\%$ noise, flat around the $k=16$
used in the paper. All quantities normalised by the characteristic strain rate
$U/L$. Regeneration script:
\texttt{figures/gen\_fig\_divergence\_validation.py}.}
\label{fig:divergence_validation}
\end{figure}

\subsection{Continuity is learned as a data pattern, not a physical law}
\label{sec:continuity}

A separate question, orthogonal to the identifiability thesis and reported here as a secondary negative control rather than a central claim, is whether the trained surrogate respects the incompressible continuity equation $\nabla\!\cdot\!\mathbf{u} = 0$ at the per-node level. We estimate the local divergence of both the predicted and ground-truth velocity fields on each test mesh by a weighted least-squares Jacobian fit on the $16$-nearest-neighbour neighbourhood (\Cref{sec:methods}), then compare the normalised mean absolute divergence $\overline{|\nabla\!\cdot\!\mathbf{u}|}\cdot L/U$, where $L$ is the median nearest-neighbour spacing and $U$ is the mean HE speed. We validated this estimator on synthetic point clouds with analytically known divergence before applying it here (\Cref{sec:supp_divergence}): on imposed-divergence fields it is unbiased (recovery slope $1.00$, intercept $<10^{-3}$ in normalised units), and on a divergence-free field it reads a small noise-controlled floor ($\approx 0.01$--$0.03$ at $0$--$2\%$ velocity noise) that is stable across neighbourhood size. The comparisons below are therefore read against that calibrated floor, not against an idealised zero.

On the VMR cohort, the predicted and true divergence residuals agree to within a factor of $\sim\!1.5$ across all four variants (\dirmagprobe{} normalised $\overline{|\nabla\!\cdot\!\mathbf{u}_{\mathrm{pred}}|} = 0.0027$ vs.\ true $0.0021$; \geoschema{} $0.0035$ vs.\ $0.0021$). This could be read as the network having internalised mass conservation as a physical constraint. The Womersley cross-domain check refutes that reading. On the analytical synthetic benchmark---where the true field is divergence-free to machine precision ($\overline{|\nabla\!\cdot\!\mathbf{u}_{\mathrm{true}}|} \approx 10^{-4}$)---the predicted divergence is $\sim\!10^2$ times larger ($\overline{|\nabla\!\cdot\!\mathbf{u}_{\mathrm{pred}}|} \approx 10^{-1}$) across every variant including \wmwith{}, for which the velocity target is supplied at input time. The same model that ``respects continuity'' on VMR therefore fails to conserve mass on a regime where the true field is exactly mass-conserving.

The apparent mass conservation on VMR is therefore a learned property of the \emph{training-data distribution}---a side-effect of fitting node-level velocities on a finite, discretisation-noise-bearing CFD output---rather than a property of the network architecture or loss. The continuity-respecting behaviour observed on VMR does not transfer to a regime where it would be most expected to hold, and is therefore best understood as pattern recognition over the training cohort, not as physics imposition. This distinction matters for downstream use: a vascular-flow surrogate whose continuity behaviour is set by training-data statistics rather than by the variational structure of the model offers no guarantee on cases outside the cohort distribution.

\subsection{Per-pathology stratified table}

{\sloppy
The full stratified results (PP@10, PP@5, angle, $\ewRMSE$,
peak-localisation, peak-magnitude-relative, WSS MAE, and WSS $R^2$)
for the strata $\{$healthy, CoA-rigid, CoA-FSI, rigid-all,
all-test$\}$ are committed in
\texttt{results/stratified/\allowbreak stratified\_test.csv}.\par}

\subsection{Peak-localisation map}

\begin{figure}[H]
\centering
\includegraphics[width=0.9\linewidth]{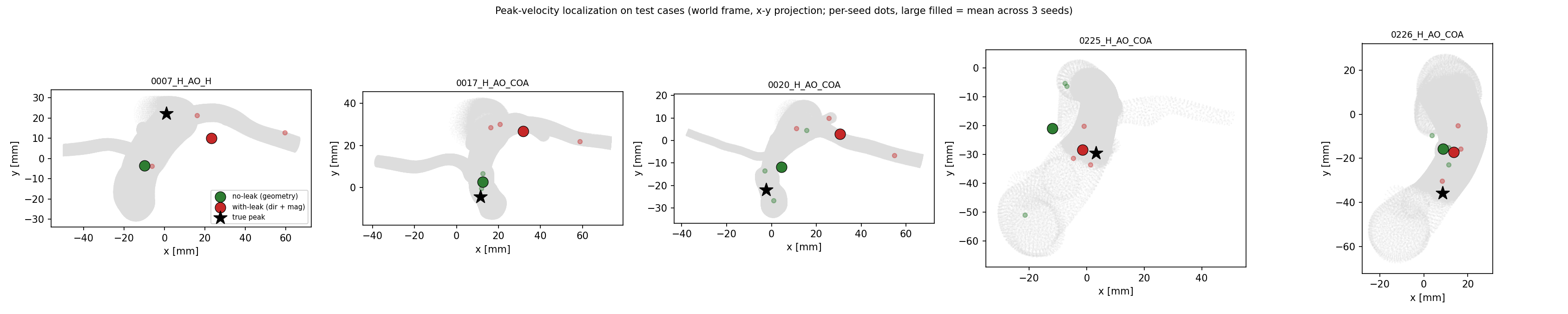}
\caption{\textbf{Predicted peak-magnitude positions overlaid on the
test-case 3-D meshes.} True peak (green star) versus per-(variant,
seed) predicted peak (coloured markers). The two FSI cases
(\texttt{0225}, \texttt{0226}) drive the peak-localisation
statistic-dependence discussed in
\Cref{sec:peakloc}.}
\label{fig:fig4}
\end{figure}

\subsection{Architecture-independence check (GraphSAGE backbone)}
\label{sec:supp_sage}

To verify that the asymmetric direction/magnitude pattern is not a property of the specific FlowGAT attention design, we retrain two of the four variants on a structurally distinct backbone: a vanilla GraphSAGE message-passing stack with mean aggregation and identical hidden width, depth, and optimiser settings, but no attention, no edge-bias, and no hard-no-slip head (\Cref{sec:supp_sage_methods}). The two variants chosen, \sagedirmagprobe{} and \sagegeo{}, bracket the ceiling and floor of the asymmetric pattern.

\Cref{tab:sage_compare} reports the comparison. SAGE reproduces the quantitative pattern: with a direction oracle probe the model achieves oracle-regime performance ($\PPdir{10} = 0.985 \pm 0.019$, $\cosgn = +0.998 \pm 0.001$, $\fracflip = 0$); without it, the model collapses to the geometry-only-class regime ($\PPdir{10} = 0.122 \pm 0.088$, $\cosgn = +0.700 \pm 0.366$, $\fracflip = 0.188 \pm 0.174$). The per-node physics diagnostics sharpen the picture: \sagedirmagprobe{} reports a \emph{higher} Spearman correlation between predicted speed and the Poiseuille radial coordinate ($\rho = +0.47$) than the corresponding FlowGAT \dirmagprobe{} run ($\rho = +0.36$), indicating that the simpler architecture is, if anything, slightly better at recovering the parabolic radial profile from a faithful direction frame.

\begin{table}[H]
\centering
\scriptsize
\setlength{\tabcolsep}{3pt}
\renewcommand{\arraystretch}{1.08}
\begin{tabularx}{\linewidth}{lXcccc}
\toprule
backbone & variant & $\PPdir{10}$ & $\cosgn$ & $\fracflip$ & angle\textsuperscript{med} ($^\circ$) \\
\midrule
FlowGAT & \dirmagprobe{} 
& $0.955 \pm 0.038$ 
& $+0.998 \pm 0.001$ 
& $0.000$ 
& $3.2 \pm 0.9$ \\

SAGE & \sagedirmagprobe{} 
& $0.985 \pm 0.019$ 
& $+0.998 \pm 0.001$ 
& $0.000$ 
& $3.1 \pm 1.1$ \\

\midrule

FlowGAT & \geoschema{} 
& $0.136 \pm 0.078$ 
& $+0.798 \pm 0.177$ 
& $0.159 \pm 0.123$ 
& $34.8 \pm 14.9$ \\

SAGE & \sagegeo{} 
& $0.122 \pm 0.088$ 
& $+0.700 \pm 0.366$ 
& $0.188 \pm 0.174$ 
& $41.5 \pm 25.5$ \\
\bottomrule
\end{tabularx}
\caption{\textbf{Architecture independence of the identifiability
limit.} A vanilla GraphSAGE backbone reproduces both the
direction oracle probe ceiling and the geometry-only floor with
comparable per-node sign-degeneracy.}
\label{tab:sage_compare}
\end{table}

\subsection{Refined geometric proxy: per-node centreline tangent}
\label{sec:supp_centerline}

A natural reading of the geometry-only failure is that the case-global PCA axis is a crude direction proxy and that the network's residual angular error tracks the geometric gap between this proxy and the true centreline tangent at each node. If that reading were correct, replacing the global PCA proxy with a \emph{local} centreline tangent should narrow the gap. We test it directly with a fifth variant, \centerlinegeo{}, which substitutes the global PCA axis with a per-node tangent computed iteratively from a medial-axis skeleton (\Cref{sec:centerline-tangent-method}): the procedure traces a centreline through the lumen by successive Voronoi-pole pruning and projects the resulting Frenet tangent back onto each interior node by nearest surface point.

The result is essentially identical to \geoschema{}. On the direction-identifiability metrics we observe $\PPdir{10} = 0.087 \pm 0.042$ for \centerlinegeo{} versus $0.136 \pm 0.078$ for \geoschema{}; $\cosgn = +0.799 \pm 0.131$ versus $+0.798 \pm 0.177$; $\fracflip = 0.136 \pm 0.076$ versus $0.159 \pm 0.123$; and median angular error $35.4 \pm 12.0^\circ$ versus $34.8 \pm 14.9^\circ$ (three seeds, five test cases). On PP@10, $\ewRMSE$, and per-pathology stratification the two variants are likewise indistinguishable within seed-to-seed scatter. Refining the geometric direction proxy from a one-axis-per-case estimate to a per-node estimate---making the geometric prior provably more faithful at every interior node---does not move the model output, the angular error distribution, or the per-node sign degeneracy. This is the counter-intuitive complement to the \dironly{} result: the network does not benefit from a better geometric direction estimate either as input or as augmented prior; what it cannot derive from the mesh at all is the per-node \emph{instantaneous sign} of the velocity, which is exactly what a direction oracle probe supplies (\Cref{lem:sign}).

\subsection{No-BC ablation variant}
\label{sec:methods_nobc}

The no-slip boundary condition in the standard FlowGAT model is enforced via a hard-constraint loss head: wall nodes (where $\mathbf{u} = \mathbf{0}$) receive a supervised auxiliary loss term that explicitly signals their zero-velocity boundary condition to the network during training. This head means that information about the boundary-condition type (no-slip vs.\ inflow) is available to the model at wall nodes and could, in principle, help anchor the orientation of interior flow via boundary-to-interior message passing. To isolate this concern we train the four-variant Cosserat-sweep ablation on a version of FlowGAT in which the no-slip-loss head is disabled: wall nodes are excluded from the supervised loss entirely (as is standard in mesh-conditioned surrogates that do not explicitly label boundary-condition types in the feature set). All other hyperparameters---hidden width, depth, attention heads, edge-bias design, AdamW learning rate, augmentation, three seeds, and train/val/test split---are held fixed. The resulting model must infer the wall boundary condition from the geometric features alone (specifically, the node-type indicator distinguishing interior from wall nodes in $\mathcal{G}(\Omega)$), which is the condition the theorem assumes for the geometry-only schema (\Cref{def:input-schema}). The no-BC run therefore provides a cleaner instantiation of the theoretical setting than the standard model. Its results are reported in the third block of \Cref{tab:e8_direction} and discussed in \Cref{sec:e8_robustness}.

\subsection{Methods for the GraphSAGE backbone}
\label{sec:supp_sage_methods}

The GraphSAGE backbone used in \Cref{sec:supp_sage} is a message-passing stack
\[
\mathbf{h}_i^{(\ell+1)} = \sigma\!\left(\mathbf{W}_\mathrm{self}^{(\ell)}\,\mathbf{h}_i^{(\ell)}
+ \mathbf{W}_\mathrm{nbr}^{(\ell)}\,\frac{1}{|\mathcal{N}(i)|}\!\!\sum_{j\in\mathcal{N}(i)}\!\!\mathbf{h}_j^{(\ell)}\right),
\]
without attention weights, without edge-bias terms, and without the hard-no-slip head. Hidden width ($128$), depth ($8$ layers), activation (SiLU), optimiser (AdamW, learning rate $10^{-4}$), augmentation, subgraph sampling, and training-epoch budget are identical to the FlowGAT runs. The parameter count drops to $\sim\!348{,}000$. We train each of the two variants (\sagedirmagprobe{}, \sagegeo{}) with three seeds.


\bibliography{refs}

\end{document}